\documentclass[a4paper,11pt]{article}

\usepackage{blindtext,titlefoot}

\usepackage{amsfonts,amsthm,amssymb,amsmath}
\usepackage[title,titletoc,toc]{appendix}
\usepackage[utf8]{inputenc}
\usepackage[affil-it]{authblk}
\usepackage{xcolor}
\usepackage{tikz}
\newcommand\COMP{\hbox{C\kern -.58em {\raise .54ex \hbox{$\scriptscriptstyle |$}}
\kern-.55em {\raise .53ex \hbox{$\scriptscriptstyle |$}} }}
\newcommand\NN{\hbox{I\kern-.2em\hbox{N}}}
\newcommand\RR{\hbox{I\kern-.2em\hbox{R}}}
\newcommand\sRR{{\it \hbox{I\kern-.2em\hbox{R}}}}
\newcommand\QQ{\hbox{I\kern-.53em\hbox{Q}}}
\newcommand\PP{\hbox{I\kern-.53em\hbox{P}}}
\newcommand\EE{\hbox{I\kern-.53em\hbox{E}}}
\newcommand\ZZ{{{\rm Z}\kern-.28em{\rm Z}}}
\newcommand\be{\begin{equation}}
\newcommand\ee{\end{equation}}
\newtheorem{theorem}{Theorem}[section]

\newtheorem{proposition}[theorem]{Proposition}
\newtheorem{remark}[theorem]{Remark}

\newtheorem{example}[theorem]{Example}
\newtheorem{lemma}[theorem]{Lemma}

\newtheorem{corollary}[theorem]{Corollary}
\newtheorem{definition}[theorem]{Definition}

\newcommand{\E}{\mathbb E}

\makeatletter
\newcommand*\bigcdot{\mathpalette\bigcdot@{.5}}
\newcommand*\bigcdot@[2]{\mathbin{\vcenter{\hbox{\scalebox{#2}{$\m@th#1\bullet$}}}}}
\makeatother
\newcommand{\is}{\bigcdot }

\def \Lbrack {[\![}
\def \Rbrack {]\!]}

\numberwithin{equation}{section}

\newcommand{\overbar}[1]{\mkern 3.5mu\overline{\mkern-3.5mu#1\mkern-1.5mu}\mkern 1.5mu}
\DeclareMathOperator*{\loc}{loc}
\DeclareMathOperator*{\Var}{Var}

\begin{document}
%\title{Mortality Risk-Minimization and Optional Martingale Representation}
\title{A martingale representation theorem and valuation of defaultable securities}

\author{Tahir Choulli\\ University of Alberta \and Catherine Daveloose and Mich\`ele Vanmaele\\Ghent University}

%\author[1]{Tahir Choulli}
%\author[2]{Catherine Daveloose}
%\author[2]{Mich\`ele Vanmaele}
%\affil[1]{\small{Department of Mathematical and Statistical Sciences, University of Alberta, 632 Central Academic Building,
%Edmonton, AB T6G 2G1, Canada, tchoulli@ualberta.ca}}
%\affil[2]{Department of Applied Mathematics, Computer Science, and Statistics, Ghent University,  Krijgslaan 281-S9, 9000 Gent, Belgium,  Catherine.Daveloose@UGent.be, Michele.Vanmaele@UGent.be}

% \date{}

\maketitle\unmarkedfntext{
This research is supported by NSERC (through grant NSERC RGPIN04987)\\
The authors wish to thank A. Aksamit, J. Deng, V. Galvani, M. Jeanblanc, M. Rutkowski, Th. Schmidt, M. Schweizer, and the participants of the 11th Bachelier (Januray 2017) for helpful advices, discussions, comments and/or suggestions.\\
Tahir Choulli is thankful to Mich\`ele Vanmaele and Department of Applied Mathematics, Computer Science, and Statistics (Ghent University),  for their invitation and  hospitality, where this work started.\\
The authors are very grateful to two anonymous referees and an anonymous AE for their pertinent suggestions and comments that shaped this final version. Possible errors are our sole responsibility.\\
Address correspondence to Tahir Choulli, Department of Mathematical and Statistical Sciences, University of Alberta, 632 Central Academic Building,
Edmonton, Canada; e-mail: tchoulli@ualberta.ca}

\begin{abstract}
We consider a market model where there are two levels of information. The public information generated by the financial assets,
and a larger flow of information that contains additional knowledge about a random time. This random time can represent many economic and financial settings, such as the default time of a firm for credit risk, and the death time of an insured  for life insurance.
By using the expansion of filtration, the random time uncertainty and its entailed risk are fully considered without any mathematical restriction.
In this context with no model's specification for the random time, the main challenge lies in finding the dynamics and the structures for the value processes of defaultable or mortality and/or longevity securities which are vital for the insurance securitization. To overcome this obstacle, we elaborate our optional martingale representation
results, which state that any martingale in the large filtration stopped at the random time can be decomposed into precise and unique orthogonal local martingales (i.e. local martingales whose product remains a local martingale).
This constitutes our first and probably the principal contribution. Even though the driving motivation for this representation
resides in credit risk theory, our results are applicable to several other financial and economics contexts, such as life insurance and financial markets with random horizon. Thanks to this optional representation, we decompose any defaultable or mortality and/or longevity liability into the sum of ``non-correlated" risks using a risk basis. This constitutes our second contribution.  \\


\end{abstract}

\section{Introduction}
This paper considers a financial framework that is defined by an initial market, which is characterized by its flow of information $\mathbb F$ and its underlying traded assets $S$, and a random time that might not be seen through $\mathbb F$ when it occurs. In this setting, our principal goals reside in classifying the risk, defining a risk basis, and then representing any general and complex risk on this basis afterwards. Our results are clearly motivated and are applicable to credit risk theory, life insurance, and random horizon in finance.  While the literature about random horizon is timid and mostly in economics with discrete market models, the literature about credit risk, for which $\tau$ represents the default time, is very extensive. Among this literature, we refer the reader to  \cite{bieleckirutkowski02, blanchetjeanblanc04, guimlichscmidt, jiaoli}, and the references therein to cite few. The role and the importance of our martingale classification and representation sounds clear from this literature. For life insurance, the role of our results  becomes important due to the newly direction in mortality and/or longevity risk management called securitization. However, this role does not seem clear nor straightforward, and requires more details. Thus,  for the reader's convenience, the rest of this paragraph elaborates about this chapter of projected applications for our results. \\ 

Life insurance companies and pension funds face two main type of risks: financial risk and mortality or longevity risk.
 The financial risk is related to the investment in risky assets, while the mortality risk follows from the uncertainty of
death time and can be split into a systematic and an unsystematic part. For more details about this we refer to \cite{dhalmelchiormoller08, dhalmoller06} and the references therein.  Longevity risk refers to the risk that the realized future mortality
 trend exceeds current assumptions. This risk beside the systematic mortality risk cannot be diversified by
  increasing the size of the portfolio. Recently there has been an upsurge interest in transferring such illiquid risks into
  financial markets allowing risk pooling and risk transfer for many retail products. This process is known as securitization,
  and it started for the pure insurance risk in mid-1990s through insurance linked securitization and the catastrophe bond market.
  The initial risk securitization was the Swiss Re Vita Capital issue in December 2003. In \cite{blakeburrows01}, see also  \cite{barrieuetal09,blakeetal08b}  and the references therein, the authors were
  the first to advocate the use of mortality-linked securities to transfer longevity risk to capital market. One of the key challenge in this mortality securitization lies in finding the prices and their dynamics
   for the death securities that will be used  in this securitization such as longevity bonds. These prices obviously depend heavily on the mortality model used and the method used to price those
   securities. Since the Lee-Carter model of \cite{leecarter92}, there were many suggestions for mortality modelling. These models  can be classified into two main groups, depending whether the obtained model was inspired from credit risk modelling,
  or interest rate modelling. Many models assume that the paths of the conditional survival probability is decreasing in time. This was severally criticized by \cite{barbarin09},
  where the author proposes to model longevity bonds {\it \`a la} Heath-Jarrow-Morton. Recently, in \cite{friedbergwebb07} (see also \cite{cairnsetal06, bauer10} for related discussion),  the authors use the CAPM and the CCPAM to price longevity bonds, and concluded that this pricing is not accurate with the reality and suggest that there might be a kind of ``mortality premium puzzle" {\it \`a la } Mehra and Prescott \cite{Mehra1985}. While this mortality premium puzzle might exist, the ``poor and/or bad" specification
  of the model for the mortality plays an important role in getting those wrong prices for longevity bonds. Thus, naturally, one can ask whether the dynamics of the longevity bond's price process can be described without mortality specification. Or in other words, one ask the following
 
  \begin{equation}\label{Question1}
\mbox{What are the dynamics of a defaultable security's value process for an arbitrary $\tau$?}
\end{equation}

\subsection{Our main objectives and the related literature}
 To describe precisely and mathematically our main aims in this paper, we need some notation.  Throughout the whole paper, we consider given the financial market model described mathematically by the tuple $\left(\Omega, {\cal G}, {\mathbb F},S, P\right)$.  Herein, the filtered probability space $\left(\Omega,{\cal G},  {\mathbb  F}=({\cal F}_t)_{t\geq 0}, P\right)$ satisfies the usual condition (i.e. filtration is complete and  right continuous) with ${\cal F}_t\subset{\cal G}$, and $S$ is an $\mathbb F$-semimartingale representing the discounted price process of $d$ risky assets. The random time (default or mortality) is modelled with $\tau$, which is mathematically an arbitrary $[0,+\infty]$-valued random variable.  The flow of information generated by the public flow $\mathbb F$ and the random time $\tau$ will be denoted by $\mathbb G$, where the relationship between the three components $\mathbb F$, $\tau$ and $\mathbb G$ will be specified in the next section. \\
 
 Thus, our goals can be summarized into two main objectives that are intimately related. The first objective resides in giving a precise and rigorous answer to (\ref{Question1}). Up to our knowledge, all the existing literature about default or mortality and/or longevity assumes a specific model for $\tau$  and derives the dynamics of the value processes for the securities under consideration accordingly.  Our second  objective is the classification of risk into three categories with their mathematical modelling, and the elaboration of the following relationship.  
\begin{equation}\label{RiskDecomposition}
{\mathbb  G}\mbox{-Risk up to}\ \tau={\cal R}\Bigl(\mbox{PFR},\mbox{PDR}^1,...,\mbox{PDR}^k,\mbox{CR}^1,...,\mbox{CR}^l\Bigr). \end{equation}

 In this equation, the dummies PFR, PD and CR refer to ``pure" financial risk, pure default risk, and correlation risk --intrinsic to the correlation between the financial market and the default-- respectively. The function ${\cal R}$ is the functional that connects all three types of risk to the risk in $\mathbb G$ up to $\tau$. Thanks to arbitrage theory, a risk can be mathematically assimilated to a martingale. Thus, in this spirit, the equation (\ref{RiskDecomposition}) can be re-written using martingale theory as follows. For any martingale under $\mathbb G$ stopped at $\tau$, $M^{\mathbb G}$, we aim the following 
 \begin{equation}\label{GmartingaleRepresentation}
M^{\mathbb G}=M^{\mbox{(pf)}}+M^{\mbox{(pd)}}_1+...+M^{\mbox{(pd)}}_k+M^{\mbox{(cr)}}_1+...+M^{\mbox{(cr)}}_l. \end{equation}
All the terms in the RHS of the above equation are $\mathbb G$-local martingales, preferably mutually orthogonal\footnote{two local martingales are orthogonal if their product is also a local martingale}, representing pure financial risk, pure default risks and correlated risks respectively.  This representation goes back to \cite{azemaetal93}, where the authors established a similar representation in the Brownian setting and when  $\tau$ is the end of an $\mathbb F$-predictable set avoiding $\mathbb F$-stopping times. These two conditions on the pair $(\mathbb F, \tau)$ are vital in their analysis and proofs. Motivated by credit risk theory, \cite{blanchetjeanblanc04} extended \cite{azemaetal93} to the case where the triplet $(\mathbb F, \tau, M^{\mathbb G})$ satisfies the following two assumptions:\\
\begin{equation}\label{assumption1}
\mbox{Either}\ \tau\ \mbox{avoids}\ \mathbb F\mbox{-stopping times or all}\ \mathbb F\mbox{-martingales are continuous,}
\end{equation}
 and 
\begin{equation}\label{assumption2} 
M^{\mathbb G}\  \mbox{is given by}\ M^{\mathbb G}_t:=E(h_{\tau}\  \big|\ {\cal G}_t)\  \mbox{where}\ h\ \mbox{is}\ \mathbb F\mbox{-predictable with suitable integrability.}
\end{equation}
 It is worth mentioning, as the authors themselves realized it,  that the representation of \cite{blanchetjeanblanc04} fails as soon as  (\ref{assumption1}) or  (\ref{assumption2})  is violated. It is clear that for the popular  and simple discrete time market models the assumption (\ref{assumption1}) fails. Furthermore for most models in insurance (if not all), Poisson process is an important component in the modelling,  and hence for these models the second part of assumption (\ref{assumption1}) fails, while its first part can be viewed as a kind of ``independence" assumption between the random time $\tau$ and the financial market (i.e. the pair $(\mathbb F, S)$). In \cite{gerbershiu13} and the references therein, the authors treat many death-related claims and liabilities in (life) insurance whose payoff process $h$ fails (\ref{assumption2}).  Our representation (\ref{GmartingaleRepresentation}) is elaborated under no assumption of any kind, and hence leading to {\it new martingales} and an innovative mathematical modelling for risk. 

\subsection{Our financial  and mathematical achievements}
In our view, it is highly important to mention that {\it our results }--even though they are motivated by (and applied to)  credit risk theory -- {\it are quite universal in the sense that they are applicable to more broader financial and economics domains}. Among these, we cite life insurance in general and in particular mortality/longevity risk and its securitization, and  markets with random horizon,..., etcetera. Our main contributions can be summarized into two blocks that are intimately related to each other and are our answers to the aforementioned two main objectives of the previous subsection. First of all, we mathematically define the pure default risk by introducing the pure default (local) martingales, and we classify them into two types that are orthogonal to each other. Then we represent any $\mathbb G$-martingale, stopped at $\tau$, as the sum of three orthogonal local martingales. Two of these are of the first type and the second type of pure default local martingales, while the third local martingale is the sum of two local martingales representing the ``pure" financial risk and the correlation risk between the financial market and the default.  This innovative contribution answers fully and explicitly (\ref{RiskDecomposition}). For a chosen martingale measure in the large filtration, we describe the dynamics of the discounted price processes of some defaultable securities. This answers (\ref{Question1}), and lays down --in our view-- the main ``philosophical" idea behind the stochastic structures of defaultable (or mortality) securities' valuation.   \\

\noindent This paper contains three sections, including the current section, and an appendix. The aim of Sect. \ref{section2}  lies in introducing and developing pure default (local) martingales, and elaborating the complete and general optional martingale representation as well.  This section represents the principal innovation of the paper.  The third section addresses the dynamics of the value processes of some defaultable securities. For the sake of easy exposition, the proof of some results are relegated to the  appendix.
%%%%%%%%%%%%%%%%%%%%%%%%%%%%%%%%%%%%%%%%%%%%%%%%%%%%%%%%%%%%%%%%%%%%%%%%%%%
\section{Decomposition of $\mathbb G$-martingales stopped at $\tau$}\label{section2}
%\label{secRepresentation}
This section provides the complete, explicit, and general form for (\ref{GmartingaleRepresentation}) or equivalently (\ref{RiskDecomposition}). To this end, we need to define the relationship between $(\mathbb F, \tau)$ and $\mathbb G$, and recall notation that will be used throughout the rest of the paper. 
Throughout the paper, we denote
\begin{equation}\label{ProcessDandG}
D := I_{\Lbrack\tau,+\infty\Lbrack},\ \ \ \mathbb G:=({\cal G}_t)_{t\geq 0},\ \ \ {\cal G}_t=
\cap_{s>0}\left({\cal F}_{s+t}\vee\sigma\left(D_{u},\ u\leq s+t\right)\right).\end{equation}

For any filtration  $\mathbb H\in \{\mathbb F,\mathbb G\}$, we denote by ${\cal A}(\mathbb H)$ (respectively ${\cal M}(\mathbb H)$) the set
of $\mathbb H$-adapted processes with $\mathbb H$-integrable variation (respectively that are $\mathbb H$-uniformly integrable martingales).
For any process $X$,  $^{o,\mathbb H}X$  (respectively $^{p,\mathbb H}X$)  is the
$\mathbb H$-optional (respectively $\mathbb H$-predictable) projection of $X$ when it exists. For a process with finite variation $V$,  the process $V^{o,\mathbb H}$ (respectively $V^{p,\mathbb H}$) represents its dual $\mathbb H$-optional (respectively $\mathbb H$-predictable) projection when it exists. For a filtration $\mathbb H$, ${\cal O}(\mathbb H)$, ${\cal P}(\mathbb H)$ and  Prog$(\mathbb H)$ denote the $\mathbb H$-optional, the $\mathbb H$-predictable and the $\mathbb H$-progressive $\sigma$-fields  respectively on $\Omega\times[0,+\infty[$. For an $\mathbb H$-semimartingale $X$, we denote by $L(X,\mathbb H)$ the set of all $X$-integrable processes in Ito's sense. Furthermore, when $H\in L(X,\mathbb H)$, the resulting integral is a one dimensional $\mathbb H$-semimartingale denoted by $H\is X:=\int_0^{\cdot} H_u dX_u$. If ${\cal C}(\mathbb H)$ 
is a set of processes that are $\mathbb H$-adapted,
then ${\cal C}_{\loc}(\mathbb H)$ --{\bf except when it is stated otherwise}-- is the set of processes, $X$,
for which there exists a sequence of $\mathbb H$-stopping times,
$(T_n)_{n\geq 1}$, that increases to infinity and $X^{T_n}$ belongs to ${\cal C}(\mathbb H)$, for each $n\geq 1$.
 Throughout the paper, we consider the following processes
\begin{equation}\label{GGtildem}
G_t :=\  ^{o,\mathbb F}(I_{\Lbrack0,\tau\Lbrack})_t=P(\tau > t | {\cal F}_t),\ \ \ \ \ \ \ \ \ \widetilde{G}_t := \ ^{o,\mathbb F}(I_{\Lbrack0,\tau\Rbrack})_t=P(\tau \ge t | {\cal F}_t),
\quad \mbox{ and } \quad \ m := G + D^{o,\mathbb F}.
\end{equation}
 Both $G$ and $\widetilde G$ are known as Az\'ema supermartingales ($G$ is right-continuous with left limits, while in general $\widetilde G$ has right and left limits only), and $m$ is a BMO $\mathbb F$-martingale. For more details about these, we refer the reader to  \cite[paragraph 74, Chapitre XX]{dellacheriemeyer92}. We ends this preliminary paragraph by recalling the definition of orthogonality between local martingales.
 
 \begin{definition}\label{Orthogonality} Let $M$ and $N$ be two $\mathbb H$-local martingales. Then $M$ is  said to be orthogonal to $N$ whenever $MN$ is also an $\mathbb H$-local martingale, or equivalently $[M,N]$ is an $\mathbb H$-local martingale. 
 \end{definition}
 %%%%%%%%%%%%%%%%%%%%%%%%%%%%%%%%%%%%%%%%%%%%%%%%%%%%%%%%%%%%%%%%%
 \subsection{Compensated martingales}\label{subsection2.1}
This subsection introduces and analyzes a new class of compensated martingales, and mathematically defines the {\it pure default (local) martingales}.
\begin{definition}\label{PureMortalityMartingales}
We call a {\it pure default martingale} (respectively pure default local martingale) any $\mathbb G$-martingale
(respectively  $\mathbb G$-local martingale) $M^{\mathbb G}$ satisfying the following.\\
{\rm{(a)}} $M^{\mathbb G}$ is stopped at $\tau$ (i.e. $ M^{\mathbb G}=\left(M^{\mathbb G}\right)^{\tau}$).\\
{\rm{(b)}} $M^{\mathbb G}$  is orthogonal to any $\mathbb F$-locally bounded local martingale
(i.e. $[M^{\mathbb G}, M]$ is a $\mathbb G$-local martingale for any $\mathbb F$-locally bounded local martingale $M$).
\end{definition}

As we mentioned in the introduction, the results of this section (Section \ref{section2}) are more general and applicable to various financial and economics areas such as random horizon in finance, and life insurance (mortality and longevity risks) which has strong similarities with credit risk theory. Thus, similarly as in the definition above, one can  define the pure mortality, or pure horizon, (local) martingales.\\

In virtue of this definition, a pure default martingale is a martingale that is intimately related to the uncertainty in $\tau$ that can not be seen through $\mathbb F$.
Thus, naturally, one can ask whether the $\mathbb G$-martingale in the Doob-Meyer decomposition of $D$ defined in (\ref{ProcessDandG}), which is given by
 \begin{equation} \label{processNGbar}
\overbar{N}^{\mathbb G} :=D -G_-^{-1} I_{\Rbrack 0,\tau\Rbrack} \is D^{p,\mathbb F},
\end{equation}
is a pure default martingale or not? In general, the answer is negative.
This follows from the fact that for any bounded $\mathbb F$-local martingale $M$, the process
 $[M,\overline{N}^{\mathbb G}]=\Delta M\is D -(\Delta M)G_-^{-1} I_{\Rbrack 0,\tau\Rbrack} \is D^{p,\mathbb F}$
 might not be a martingale for some pair $(M,\tau)$. Thus, the following challenging question arises.
\begin{eqnarray*}\label{QPMM1}
 \mbox{How can we construct and/or identify pure default (local) martingales ?}
\end{eqnarray*}
It is clear, from this short discussion, that the compensation {\it \`a la} Doob-Meyer does not produce pure default martingales from $\tau$. Thus, we propose below an other compensation procedure, which leads to a {\it new class} of $\mathbb G$-martingales.

\begin{theorem}\label{NGmartingaleproperties} The following process,
\begin{equation} \label{processNG}
N^{\mathbb G}:=D - \widetilde{G}^{-1} I_{\Rbrack 0,\tau\Rbrack} \is D^{o,\mathbb  F},
\end{equation}
 is a $\mathbb G$-martingale with integrable variation. 
\end{theorem}

\begin{proof} It is clear that $N^{\mathbb G}$ is a RCLL and $\mathbb G$-adapted process satisfying
\[
	\max\left(\E [\Var( N^{\mathbb G})_{\infty}],\E\Big[\sup_{t\ge 0}|N_t^{\mathbb G}|\Big]\right)\leq
 \E[  D_\infty] + \E\Big[\widetilde{G}^{-1} \ ^{o,{\mathbb F}} (I_{\Lbrack 0,\tau\Rbrack}) I_{\{\widetilde{G}>0\}}\is D_\infty\Big]
 = 2 P(\tau <+\infty)\leq 2.
\]
Thus, $N^{\mathbb G}$ has an integrable variation. For any $\mathbb F$-stopping time $\sigma$, we derive
\begin{eqnarray}
\E[ N^{\mathbb G}_{\sigma}]& = \E\Big[ D_{\sigma} - \widetilde{G}^{-1} I_{\Lbrack 0,\tau\Rbrack} \is D_{\sigma}^{o,\mathbb F}\Big]
 = \E[ D_{\sigma}] - \E\Big[\widetilde{G}^{-1}I_{\{\widetilde{G}>0\}} \ ^{o,\mathbb F}(I_{\Lbrack 0,\tau\Rbrack})
 \is D_{\sigma}\Big]\nonumber\\
&=\E[D_{\sigma}] - \E\Big[ I_{\{\widetilde{G}>0\}}  \is D_{\sigma}\Big]	=0 .\label{mgG4NG}
\end{eqnarray}
The last equality follows from $I_{\{\widetilde{G}>0\}}  \is D\equiv D$  since $\widetilde G_{\tau}>0$ \(P\)-a.s.
on $\{\tau<+\infty\}$, which follows directly from \cite[Chapitre IV, Lemma 4.3]{jeulin80}. Therefore, the proof follows immediately from a
 combination of (\ref{mgG4NG}) and the fact that for any $\mathbb G$-stopping time, $\sigma^{\mathbb G}$, there exists
  an $\mathbb F$-stopping time $\sigma^{\mathbb F}$ such that 
  \begin{equation}\label{tauGtauF}
\sigma^{\mathbb G}\wedge \tau=\sigma^{\mathbb F}\wedge \tau,\ \ \ P\text{-a.s}.\end{equation}
For this fact, we refer to \cite[ Chapter XX, paragraph 75, assertion (b)]{dellacheriemeyer92} and  \cite[Proposition B.2-(b)]{aksamitetal15}. This ends the proof of the theorem.\end{proof}

The following characterizes, in general, the situation 
where the Doob-Meyer's compensation generates pure default (local) martingales.
%%%%%%
%%%%
%%%%%%%%%%%%%%%%%%%%%%%%%%%%%%%%%%%%%%%%

\begin{proposition}\label{NbarG-NG} Consider $\overline{N}^{\mathbb G}$ and $N^{\mathbb G}$ defined in
(\ref{processNGbar}) and (\ref{processNG}).
 Then the following are equivalent.\\
 {\rm {(a)}} $\overline{N}^{\mathbb G}$ is a pure default martingale.\\
 {\rm {(b)}} $\overline{N}^{\mathbb G}$ and $N^{\mathbb G}$ coincide.\\
 {\rm {(c)}} The two processes $^{p,\mathbb F}(G){\widetilde G}$ and $G_{-}G$ are indistinguishable.
 \end{proposition}

\begin{proof} This proof is divided into two steps where we deal with (b)$\Longleftrightarrow$(c) and (a)$\Longleftrightarrow$(b) respectively. \\
{\bf Step 1:} Here we prove (b)$\Longleftrightarrow$(c). Remark that $\overline{N}^{\mathbb G}\equiv N^{\mathbb G}$ iff 
${\widetilde G}^{-1}I_{\Rbrack0,\tau\Rbrack}\is D^{o,\mathbb F}= G^{-1}_{-}I_{\Rbrack0,\tau\Rbrack}\is D^{p,\mathbb F},$ and this equality is equivalent to 
${\widetilde G}^{-1}I_{\Rbrack0,\tau\Rbrack}\Delta D^{o,\mathbb F}= G^{-1}_{-}I_{\Rbrack0,\tau\Rbrack}\Delta D^{p,\mathbb F}$. Since ${\Rbrack0,\tau\Rbrack}\subset\{G_{-}>0\}$, by taking the $\mathbb F$-optional projection, it is easy to conclude that the latter equality is equivalent to 
${G}_{-}\Delta D^{o,\mathbb F}={\widetilde  G}\Delta D^{p,\mathbb F}$. It is clear that, in turn, this obtained equality is equivalent to assertion (c), due to 
$\Delta D^{o,\mathbb F}=\widetilde G-G$ and $\Delta D^{p,\mathbb F}=G_{-}-\ ^{p,\mathbb F}(G)$. This ends the proof of (b)$\Longleftrightarrow$(c).\\
{\bf Step 2:} This part proves (a)$\Longleftrightarrow$(b). To this end, we start proving that $N^{\mathbb G}$ is in fact a pure default martingale. Let $M$ be locally bounded $\mathbb F$-local martingale, and $\sigma$ is an $\mathbb F$-stopping time. Then we get 
\begin{align*}
\E[M, N^{\mathbb G}]_{\sigma}&=\E[(\Delta M\is N^{\mathbb G})_{\sigma}]=\E\Bigl[(\Delta M\is D)_{\sigma}-{{\Delta M}\over{\widetilde G}} I_{\Lbrack 0,\tau\Rbrack} \is D_{\sigma}^{o,\mathbb F}\Big]\\
&=\E\Bigl[(\Delta M\is D^{o,\mathbb F})_{\sigma}-{{\Delta M}\over{\widetilde G}}I_{\{\widetilde G>0\}} \ ^{o,\mathbb F}(I_{\Lbrack 0,\tau\Rbrack}) \is D_{\sigma}^{o,\mathbb F}\Big]=\E\Bigl[(\Delta M I_{\{\widetilde G=0\}}\is D^{o,\mathbb F})_{\sigma}\Big]=0,
\end{align*}
where the last equality follows from the fact that $I_{\{\widetilde G=0\}}\is D^{o,\mathbb F}\equiv 0$ which is equivalent to $I_{\{\widetilde G=0\}}\is D\equiv 0$, or equivalently $\widetilde G_{\tau}>0$ $P$-a.s. on $\{\tau<+\infty\}$. The above equality combined with (\ref{tauGtauF}) prove that $[M, N^{\mathbb G}]$ is a $\mathbb G$-martingale. Thus, we conclude that $N^{\mathbb G}$ is a pure default martingale, and the proof of (b)$\Longrightarrow$(a) follows. The rest of the proof focuses on the converse. Remark that $M:=G_{-}\is D^{o,\mathbb F}-{\widetilde G}\is D^{p,\mathbb F}$ is an $\mathbb F$-local martingale with bounded jumps (hence it is locally bounded), and
$${\overline N}^{\mathbb G}-N^{\mathbb G}=(G_{-}\widetilde G)^{-1}I_{\Rbrack0,\tau\Rbrack}\is M.$$
Thus, ${\overline N}^{\mathbb G}$ is a pure default martingale iff ${\overline N}^{\mathbb G}-N^{\mathbb G}$ is also a pure default martingale. Hence, 
$$G_{-}\widetilde G\is [{\overline N}^{\mathbb G}-N^{\mathbb G}, {\overline N}^{\mathbb G}-N^{\mathbb G}]=[{\overline N}^{\mathbb G}-N^{\mathbb G}, M]\in {\cal M}_{loc}(\mathbb G),$$ or equivalently ${\overline N}^{\mathbb G}\equiv N^{\mathbb G}$ due to ${\Rbrack0,\tau\Rbrack}\subset\{G_{-}>0\}\cap\{\widetilde G>0\}$.  This proves the proposition.
\end{proof}

\noindent  Below, we discuss more particular and practical cases where we compare $N^{\mathbb G}$ and ${\overline{N}^{\mathbb G}}$.

\begin{corollary}\label{corollary4theoremM2} Consider $N^{\mathbb G}$ and ${\overbar{N}^{\mathbb G}}$ defined in (\ref{processNG}) and
(\ref{processNGbar}). Then the following assertions hold.\\
{\rm {(a)}} Suppose that $\tau$ is an $\mathbb F$-stopping time. Then $ N^{\mathbb G}\equiv 0$,
 while $\overline{N}^{\mathbb G}=I_{\Lbrack\tau,+\infty\Lbrack}-\Bigl((I_{\Lbrack\tau,+\infty\Lbrack})^{p,\mathbb F}\Bigr)^{\tau}$. As a result, in this case, 
 $\overline{N}^{\mathbb G}$ coincides with $N^{\mathbb G}$ if and only if $\tau$ is predictable.\\
{\rm{(b)}} The following conditions are all sufficient for $N^{\mathbb G}$ to coincide with ${\overline{N}^{\mathbb G}}$.\\
{\rm {(b.1)}} $\tau$ avoids \(\mathbb F\)-stopping times (i.e.\ for any \(\mathbb F\)-stopping time \(\theta\) it holds that
 \(P(\tau=\theta<+\infty)=0\)),\\
{\rm {(b.2)}} all \(\mathbb F\)-martingales are continuous,\\
{\rm {(b.3)}} $\tau$ is independent of ${\cal F}_{\infty}:=\sigma(\cup_{t\geq 0}{\cal F}_t)$.
 \end{corollary}

\begin{proof} Assertion (a) is obvious and will be omitted. Thus the rest of the proof focuses on proving assertion (b) in three parts, where we prove that each of the conditions (b.i), $i=1,2,3$, is sufficient.\\
 {\bf Part 1:} Suppose that $\tau$ avoids $\mathbb F$-stopping times. Then $\widetilde G=G$ which is equivalent to the continuity of
 $D^{o,\mathbb F}$. Thus, $D^{o,\mathbb F}=D^{p,\mathbb F}$ and
 ${\widetilde G}^{-1}I_{\Rbrack0,\tau\Rbrack}\is D^{o,\mathbb F}={G}_{-}^{-1}I_{\Rbrack0,\tau\Rbrack}\is D^{p,\mathbb F}$.
  This proves that $N^{\mathbb G}=\overline{N}^{\mathbb G}$.\\
{\bf Part 2:} Suppose that all \(\mathbb F\)-martingales are continuous. Then  $0=\Delta m=\widetilde G-G_{-}$, and the pure jump $\mathbb F$-martingale
$D^{o,\mathbb F}-D^{p,\mathbb F}$ is null. Hence, $N^{\mathbb G}$ and ${\overline N}^{\mathbb G}$ coincide in this case.\\
{\bf Part 3:} Suppose that $\tau$ is independent of ${\cal F}_{\infty}$. Then $\widetilde G$, $G_{-}$ and $G$ are deterministic.
As a consequence
we get $\widetilde G=G_{-}$ and $D^{o,\mathbb F}=D^{p,\mathbb F}$. Hence $N^{\mathbb G}=\overline{N}^{\mathbb G}$, and the proof of the corollary is completed.
\end{proof}

Even though Corollary \ref{corollary4theoremM2}-(a) is simple, it explains the main difference between the roles of $N^{\mathbb G}$ and ${\overline{N}^{\mathbb G}}$. In fact, it says that $N^{\mathbb G}$ ``measures" the extra randomness in $\tau$ that is not in $\mathbb F$, while  ${\overline{N}^{\mathbb G}}$ cannot tell which randomness comes (or that does not come) from $\mathbb F$. Hence, $N^{\mathbb G}$ is more suitable for singling out the different sources of risk. This is crucial for an efficient risk management. 

\subsection{Examples of random times}
This subsection illustrates the difference between  $N^{\mathbb G}$ and $\overline{N}^{\mathbb G}$ on practical examples of random times. 

\begin{example}\label{example1} Suppose $N$ is a Poisson process with intensity one,
$\mathbb F$ be the right continuous and complete filtration generated by $N$, and $(T_n)_{n\geq 1}$ be the sequence of $\mathbb F$-stopping times
given by
$$T_n:=\inf\{t\geq 0\ \big|\ N_t \geq n\},\ \ n\geq 1.$$
 Let $\alpha\in (0,1)$, and put
$$\tau:=\alpha T_1+(1-\alpha)T_2.$$
Thanks to \cite[Proposition 5.3] {aksamitetal15}, it is clear that $\tau$ fulfills assumption (b.1) of Corollary \ref{corollary4theoremM2}-(b), 
while $\mathbb F$ violates assumption (b.2). Thus, in this case, we conclude that $N^{\mathbb G}=\overline{N}^{\mathbb G}$.\end{example}

\begin{example}\label{example2} Consider the triplet $(N, (T_n)_{n\geq 1},\mathbb F)$ defined in the previous example, and put
$$\tau:=aT_2\wedge T_1,$$
with $a\in (0,1)$. Then, it is clear that the pair $(\tau,\mathbb F)$
violates all the three assumptions (b.1)-(b.3) of Corollary \ref{corollary4theoremM2}-(b). To show this fact, thanks to \cite{aksamitetal15}, we have
$$G_t=e^{-\beta t}(\beta t+1)I_{\{t<T_1\}},\ \ \ {\widetilde G}_t=G_{t-}=e^{-\beta t}(\beta t+1)I_{\{t\leq T_1\}},\ \ \ \  \beta:=a^{-1}-1.$$
Thus, we deduce that $\Delta D^{o,\mathbb F}={\widetilde G}-{G}=e^{-\beta T_1}(\beta T_1+1)I_{\Lbrack T_1\Rbrack}$.
This implies that $\tau$ does not avoid $\mathbb F$-stopping times, and
 $D^{o,\mathbb F}-e^{-\beta T_1}(\beta T_1+1)I_{\Lbrack T_1,+\infty\Lbrack}$ is a continuous process with finite variation.
 Furthermore, the $\mathbb G$-martingale
 $$N^{\mathbb G}-\overline{N}^{\mathbb G}=-I_{\Rbrack0,\tau\Rbrack}\is H^{(1)},\ \ \mbox{where}\ \  
 H^{(1)}:=I_{\Lbrack T_1,+\infty\Lbrack}(t)-t\wedge T_1,$$
 is not null.\end{example}
 
 This example shares the same features as the model considered in \cite{jiaoli}  and the references therein.  These models can be unified into a more general model as follows. 

\begin{proposition}\label{JiaLi}
Let $\sigma$ be an $\mathbb F$-stopping time, and $\tau_1$ be an arbitrary random time. Suppose that 
$$\tau=\sigma\wedge\tau_1.$$
Then the following assertions hold.\\
{\rm{(a)}}  If $\mathbb G_1$ is the smallest filtration that contains $\mathbb F$ and makes $\tau_1$ a stopping time, then $\mathbb G\subset\mathbb G_1$.\\
{\rm{(b)}}  Consider the processes $ D_1:=I_{\Lbrack\tau_1,+\infty\Lbrack}$ and 
$$N_1^{\mathbb G_1}:=D_1-{1\over{^{o,\mathbb F}(I_{\Lbrack0,\tau_1\Rbrack})}}I_{\Rbrack0,\tau_1\Rbrack}\is \left(D_1\right)^{o,\mathbb F},$$ 
i.e., the compensated martingale associated with $(\mathbb F, \tau_1)$ via (\ref{processNG}). Then  we have 
$$N^{\mathbb G}=\left(N_1^{\mathbb G_1}\right)^{\sigma-}=I_{\Rbrack 0,\sigma\Lbrack}\is N_1^{\mathbb G_1}.$$
\end{proposition}
\begin{proof} It is clear that $\tau$  is a $\mathbb G_1$-stopping time just like $\tau_1$ and $\sigma$, and assertion (a) follows immediately. To prove assertion (b), we put $D_0:=I_{\Lbrack\sigma,+\infty\Lbrack}$ and derive
$$D:=I_{\Lbrack\tau,+\infty\Lbrack}=1-(1-D_0)(1-D_1)=D_0+D_1-D_0 D_1=I_{\Lbrack0,\tau_1\Rbrack}\is D_0+I_{\Lbrack0,\sigma\Lbrack}\is D_1.$$ 
Thus, by taking the dual $\mathbb F$-optional projection on both sides, we get 
$$D^{o,\mathbb F}={\widetilde G}_1\is D_0+I_{\Lbrack0,\sigma\Lbrack}\is D_1^{o,\mathbb F},$$
where ${\widetilde G}_1:=\ ^{o,\mathbb F}(I_{\Lbrack0,\tau_1\Rbrack})$. Hence, by combining the above equality with $\widetilde G= I_{\Lbrack0,\sigma\Rbrack}{\widetilde G}_1$, the proof of assertion (b) follows. This ends the proof of the proposition.
\end{proof}

For the financial or economic interpretation of the model for $\tau$ of Proposition \ref{JiaLi}, we refer the reader to \cite{jiaoli} and the references therein.

 \begin{remark}
 {\rm {(a)}} For the family of random times of Proposition \ref{JiaLi}-(b), it is clear that $N^{\mathbb G}$ 
 might differ from ${\overline N}^{\mathbb G}$ even in the case where  $N_1^{\mathbb G_1}={\overline N_1}^{\mathbb G_1}$. In fact, for this latter situation, ${\overline N}^{\mathbb G}$  is a pure default martingale if and only if $P(\tau_1=\sigma<+\infty)=0$ (i.e. $\tau_1$ avoids $\sigma$).  This simple fact proves that the correlation between $\mathbb F$ and $\tau$ disturbs tremendously the structure of the risk, and hence one should not neglect this correlation in any sense. \\
{\rm {(b)}} It is easy to see that, in general, $\mathbb G\not=\mathbb G_1$ by taking $\tau_1=\tau_0+\sigma$ and $\tau_0$ is not an $\mathbb F$-stopping time.
\end{remark}
 
 Below we consider another interesting class of models for $\tau$, that one can encounter frequently in practice, and for which we compare the
  two processes $N^{\mathbb G}$ and $\overline{N}^{\mathbb G}$.

\begin{proposition}\label{tauInpredictablestopping}
Suppose that there exists a sequence of $\mathbb F$-stopping times, $(\theta_n)_{n\geq 1}$, satisfying
 \begin{equation}\label{modelthin4tau}
 \Lbrack\tau\Rbrack\subset \bigcup_{n=1}^{+\infty} \Lbrack\theta_n\Rbrack.\end{equation}
 Then the following assertions hold.\\
 {\rm {(a)}} If $(\theta_n)_{n\geq 1}$ are totally inaccessible, then $N^{\mathbb G}$ and $\overbar{N}^{\mathbb G}$ differ.\\
  {\rm {(b)}} Suppose that $(\theta_n)_{n\geq 1}$ are predictable. Then,
   $N^{\mathbb G}$ and $\overline{N}^{\mathbb G}$ coincide if and only if for all $n\geq 1$,
 \begin{eqnarray}\label{GeneralCondition}
 P\left(\tau=\theta_n\ \big| {\cal F}_{\theta_n}\right)P\left(\tau\geq \theta_n\ \big| {\cal F}_{\theta_n-}\right)
 =P\left(\tau=\theta_n\ \big| {\cal F}_{\theta_n-}\right)P\left(\tau\geq\theta_n\ \big| {\cal F}_{\theta_n}\right),\ \ \ \ \ \ P\mbox{-a.s.}.
 \end{eqnarray}
{\rm {(c)}}  Suppose, for all $n\geq 1$, $\theta_n$ is predictable and $(\tau=\theta_n)$ is independent of ${\cal F}_{\theta_n}$.
Then $N^{\mathbb G}=\overline{N}^{\mathbb G}$.
\end{proposition}

\begin{proof}
1) Suppose that $\theta_n$ is totally inaccessible for all $n\geq 1$. Then one can easily calculate
 $$D^{o,\mathbb F}:=\sum_{n=1}^{+\infty} P\left(\tau=\theta_n\ \big| {\cal F}_{\theta_n}\right)I_{\Lbrack\theta_n,+\infty\Lbrack},$$
 and deduce that $D^{p,\mathbb F}=\left(D^{o,\mathbb F}\right)^{p,\mathbb F}$ is continuous. Thus $D^{p,\mathbb F}\not =D^{o,\mathbb F}$, and assertion (a) is proved.\\
2) Suppose that $\theta_n$ is predictable for all $n\geq 1$. If furthermore $(\tau=\theta_n)$ is independent of ${\cal F}_{\theta_n}$
 for all $n\geq 1$, then $(\tau=\theta_n)$ is also independent of ${\cal F}_{\theta_n-}$ (since ${\cal F}_{\theta_n-}\subset {\cal F}_{\theta_n}$) and (\ref{GeneralCondition}) is clearly fulfilled in 
 this case. Thus assertion (c) follows immediately from assertion (b). To prove this latter assertion,
 it is enough to remark that (using the convention $0/0=0$)
  \begin{eqnarray*}
 {1\over{\widetilde G}}\is D^{o,\mathbb F}&:=&\sum_{n=1}^{+\infty} {{P\left(\tau=\theta_n\ \big| {\cal F}_{\theta_n}\right)}\over{
 P\left(\tau\geq\theta_n\ \big| {\cal F}_{\theta_n}\right)}}I_{\Lbrack\theta_n,+\infty\Lbrack}\ \  \ \mbox{and}\ \ \  \
 {1\over{G_{-}}}\is D^{p,\mathbb F}:=\sum_{n=1}^{+\infty}
 {{P\left(\tau=\theta_n\ \big| {\cal F}_{\theta_n-}\right)}\over{P\left(\tau\geq\theta_n\ \big| {\cal F}_{\theta_n-}\right)}}
 I_{\Lbrack\theta_n,+\infty\Lbrack}.\end{eqnarray*}
 This ends the proof of assertion (b) and the proof of the proposition as well.
 \end{proof}

\begin{remark}
 {\rm {(a)}} It is worth mentioning that discrete time market models are covered by Proposition \ref{tauInpredictablestopping}. In fact, it is enough to take  $(\theta_n)_n$ to be deterministic and nondecreasing times say $\theta_i=i$ for $i\geq 0$. Then the necessary and sufficient condition (\ref{GeneralCondition}) for $\overline{N}^{\mathbb G}$ to coincides with $N^{\mathbb G}$ becomes 
 \begin{eqnarray*}
  P\left(\tau=n\ \big| {\cal F}_{n}\right)P\left(\tau\geq n\ \big| {\cal F}_{n-1}\right)
 =P\left(\tau=n\ \big| {\cal F}_{n-1}\right)P\left(\tau\geq n\ \big| {\cal F}_{n}\right),\quad P\mbox{-a.s.}\quad \forall n\geq 1.
\end{eqnarray*} 
 {\rm {(b)}} It is easy to check that, for a model of $\tau$ satisfying (\ref{modelthin4tau}) with predictable $(\theta_n)_n$, (\ref{GeneralCondition}) is equivalent to Proposition \ref{NbarG-NG}-(c).
\end{remark}
 
A practical example for the model (\ref{modelthin4tau}) can be found in \cite{aksamitetal15} that we recall for the reader's convenience.

 \begin{example}\label{example4} Suppose that $\mathbb F$ is generated by a Poisson process $N$ with intensity one.
 Consider two real numbers $a>0$ and $\mu>1$, and set
\begin{eqnarray*}\label{tau}
\tau:=\sup\{t\geq 0:\ Y_t:=\mu t-N_t\leq a\},\ \ \ \ \ \ M_t:=N_t-t.\end{eqnarray*} It can be proved easily, see \cite{aksamitetal15}, that
$$
G=\Psi(Y-a)I_{\{ Y\geq a\}}+I_{\{Y<a\}}\ \ \ \ \mbox{and}\ \ \ \ {\widetilde G}=\Psi(Y-a)I_{\{Y>a\}}+I_{\{Y\leq a\}}.$$
Here $\Psi(u):=P\left(\sup_{t\geq 0}Y_t>u\right)$ is the ruin probability associated with the process $Y$.
This model for $\tau$ falls into the case of Proposition \ref{tauInpredictablestopping}-(c) (see \cite{aksamitetal15}), where $\theta_n$ is given by
$$\theta_n:=\inf\{t>\theta_{n-1}\ :\ Y_t=a\},\ \ \ n\geq 1,\ \ \ \theta_0=0.$$
Thus, for this model of $(\tau,\mathbb F)$, the two $\mathbb G$-martingales $N^{\mathbb G}$ and $\overline{N}^{\mathbb G}$ coincide.
\end{example}

Other practical models for the model (\ref{modelthin4tau}) of Proposition \ref{tauInpredictablestopping} , can be found in \cite{guimlichscmidt}. In the latter paper, the authors suppose that the random time representing the default time is known via $D^{p,\mathbb F}$ instead. Indeed, they assume the existence of a random time $\tau$ such that $D^{p,\mathbb F}_t=\int_0^t \Lambda_s ds+\sum_{k=1}^n \Gamma_iI_{\Lbrack U_i,+\infty\Lbrack}$ holds. Here $\Lambda$ is a nonnegative and $\mathbb F$-adapted process with $\int_0^t \vert \Lambda_s\vert ds<+\infty\ P\mbox{-a.s.}$, $(U_i)_{i=1,...,n}$ is a finite sequence of $\mathbb F$-predictable stopping times, and $\Gamma_i$ is an ${\cal F}_{U_i-}$-measurable random variable with values in $(0,1)$, for all $i=1,...,n$. The main challenging obstacle in this case lies in  proving the existence of $\tau$ associated with the given  $D^{p,\mathbb F}$. While we are not discussing this existence assumption herein, we simply notice that it has a big chance to hold if the space $(\Omega, {\cal G},P)$ is rich enough. Thus, provided the existence of such $\tau$, in virtue of Proposition \ref{tauInpredictablestopping}, we can conjecture that this $\tau$ can take various forms. In fact, one might have the form of $\tau=\tau_1\wedge\tau_2$, where $\tau_1$ satisfies (\ref{modelthin4tau}) with totally inaccessible $(\theta_n)_n$  and $\Lbrack\tau_2\Rbrack\subset\bigcup_{k=1}^n \Lbrack U_k\Rbrack$. A second form could be $\tau=\tau_1\wedge
\tau_2$, where $\tau_1$ is a random time that avoids $\mathbb F$-stopping times (also known as Cox's random time), and $\Lbrack\tau_2\Rbrack\subset\bigcup_{k=1}^n \Lbrack U_k\Rbrack$. A third form for $\tau$ could result from combining both previous forms by putting $\tau=\tau_1\wedge\tau_2\wedge\tau_3$, where $\tau_1$ and $\tau_2$ as in the second form, and $\tau_3$ satisfies (\ref{modelthin4tau}) with totally inaccessible $(\theta_n)_n$.  In virtue of the main idea of  \cite{guimlichscmidt} and the {\it optional spirit} of this current paper, one can think about considering $D^{o,\mathbb F}$ instead. In fact, one can suppose that $\tau$ is given such that  
$$D^{o,\mathbb F}=\alpha\int_0^t \Lambda_s ds+\beta\sum_{k=1}^{\infty} \Gamma_iI_{\Lbrack U_i,+\infty\Lbrack}+\gamma\sum_{k=1}^{\infty} \Delta_iI_{\Lbrack \theta_i,+\infty\Lbrack}.
$$  Here, the first and the second processes are of the same type as in  \cite{guimlichscmidt}, while for the third process $(\theta_i)_{i\geq 1}$ are $\mathbb F$-totally inaccessible stopping times, $\Delta_i$ is ${\cal F}_{\theta_i}$-measurable random variable with values in $(0,1)$, and $\alpha,\beta$ and $\gamma$ are nonnegative real numbers in $[0,1)$. 

\subsection{Classes of pure default martingales}
This subsection proposes two orthogonal classes of pure default (local) martingales  that model the pure default risks. These classes play vital roles in our martingale representation theorems of the last subsection. The following discusses the class of  {\it pure default martingales of type one}.

\begin{theorem}\label{theorem2.13} The following assertions hold.\\
{\rm{(a)}}  Let $K$ be an $\mathbb F$-optional process, which is Lebesgue-Stieltjes integrable with respect to $N^{\mathbb G}$. Then,
\begin{equation}\label{equivalencve4K}
K\is N^{\mathbb G}\in {\cal A}(\mathbb G)\quad \ \mbox{if and only if}\quad \ [K\is N^{\mathbb G}]^{1/2}
\in {\cal A}^+(\mathbb G)\quad \ \mbox{if and only if}\quad \ K\in {\mathcal{I}}^o(N^{\mathbb G},{\mathbb G}),\end{equation}

 where
\begin{eqnarray*} \label{SpaceLNG}
{\mathcal{I}}^o(N^{\mathbb G},\mathbb G) := \Big\{K\in \mathcal{O}(\mathbb F)\ \ \big|\ \
 \ \E\left[\vert{K}\vert G{\widetilde G}^{-1} I_{\{\widetilde{G}>0\}}\is D_{\infty}\right]<+\infty\Big\}.
\end{eqnarray*}
{\rm{(b)}} The following set 
 \begin{eqnarray*}\label{Mortalitytype1}
 {\cal M}^{(1)}(\mathbb G):=\left\{K\is N^{\mathbb G}\ \big|\  K\in{\mathcal{I}}^o_{loc}(N^{\mathbb G},{\mathbb G})\ \mbox{i.e.}\ K\in \mathcal{O}(\mathbb F)\ \mbox{and}\ \vert K\vert\is \Var(N^{\mathbb G})\in {\cal A}^+_{loc}(\mathbb G)\right\},
 \end{eqnarray*}
 where $\Var(N^{\mathbb G})$ is the variation process of $N^{\mathbb G}$, is a space of pure default local martingales that we call the class of pure default martingales of type one.
 \end{theorem}
 
 \begin{proof} Let $K$ be an $\mathbb F$-optional process that is Lebesgue-Stieltjes
 integrable with respect to $N^{\mathbb G}$.
 Then, using $\Delta D^{o,\mathbb F}=\widetilde G-G$, we derive
\begin{align*}
[K\is N^{\mathbb G},K\is N^{\mathbb G}]&=\sum K^2 (\Delta N^{\mathbb G})^2 =
\sum K^2 \Big( (1  - \widetilde{G}^{-1} I_{\{\widetilde{G}>0\}} \Delta D^{o,\mathbb F} ) \Delta D -
 \widetilde{G}^{-1} I_{\Rbrack 0,\tau\Lbrack} \Delta D^{o,\mathbb F} \Big)^2	 \\
& = \sum K^2 \Big( G\widetilde{G}^{-1} I_{\{\widetilde{G}>0\}} \Delta D -
\widetilde{G}^{-1} I_{\Rbrack 0,\tau\Lbrack} \Delta D^{o,\mathbb F} \Big)^2\\
&=  \sum K^2 (G\widetilde{G}^{-1})^2 I_{\{\widetilde{G}>0\}} \Delta D +
\sum K^2 \widetilde{G}^{-2} I_{\Rbrack 0,\tau\Lbrack} (\Delta D^{o,\mathbb F})^2.	\end{align*}

A combination of this together with $\sqrt{\sum \vert x\vert}\leq \sum \sqrt{\vert x\vert}$ implies that, on the one hand,
\begin{equation}\label{inequalities}
\vert K\vert G\widetilde{G}^{-1} I_{\{\widetilde{G}>0\}} \is D\leq \sqrt{K^2\is [N^{\mathbb G}, N^{\mathbb G}]}\leq
 \vert K\vert G\widetilde{G}^{-1} I_{\{\widetilde{G}>0\}} \is D+
 \vert K\vert \widetilde{G}^{-1} I_{\Rbrack 0,\tau\Lbrack}\is D^{o,\mathbb F}.
\end{equation}
On the other hand, since $G=\ ^{o,\mathbb F}(I_{\Lbrack 0,\tau\Lbrack})$, it is easy to check that
\begin{equation}\label{equiv1}
{{\vert K\vert}\over{ \widetilde{G}}} I_{\Rbrack 0,\tau\Lbrack}\is D^{o,\mathbb F}\in
{\cal A}^+(\mathbb G)\ \mbox{(resp.}\ {\cal A}^+_{\loc}(\mathbb G))\quad \mbox{iff}
\quad {{\vert K\vert G}\over{\widetilde{G}}} I_{\{\widetilde{G}>0\}} \is D\in {\cal A}^+(\mathbb G)\
 \mbox{(resp.}\ {\cal A}^+_{\loc}(\mathbb G)).\end{equation}

 Thanks again to $\Delta D^{o,\mathbb F}=\widetilde G-G$ we get $\Var(K\is N^{\mathbb G})
 =\vert K\vert G\widetilde{G}^{-1} I_{\{\widetilde{G}>0\}} \is D+ \vert K\vert \widetilde{G}^{-1}
 I_{\Rbrack 0,\tau\Lbrack}\is D^{o,\mathbb F}$. Hence, the proof of (\ref{equivalencve4K})
 follows immediately from combining this latter equality with (\ref{inequalities}) and (\ref{equiv1}). This proves assertion (a), while the rest of the proof deals with assertion (b).  
For any $\mathbb F$-stopping time \(\sigma\), and $K\in {\cal I}^o(N^{\mathbb G}, \mathbb G)$,
 due to $I_{\{\widetilde G>0\}}\is D\equiv D$, we get
\begin{align*}
\E[(K\is N^{\mathbb G})_{\sigma}] &= \E\Big[ (K \is D)_{\sigma} - (K\widetilde{G}^{-1}
I_{\Rbrack 0,\tau\Rbrack} \is D^{o,\mathbb F})_{\sigma}\Big] \\
&= \E[ (K \is D)_{\sigma}] - \E\Big[(K\widetilde{G}^{-1} \ ^{o,\mathbb F}(I_{\Lbrack 0,\tau\Rbrack})
I_{\{\widetilde{G}>0\}} \is D)_{\sigma}\Big]=0 .
\end{align*}
Thus, in virtue of (\ref{tauGtauF}), the above equality proves that $K\is N^{\mathbb G}\in {\cal M}_0(\mathbb G)$. Similarly, it is clear that $K\is N^{\mathbb G}\in {\cal M}_{0,loc}(\mathbb G)$ as soon as $K\in{\mathcal{I}}^o_{\loc}(N^{\mathbb G},{\mathbb G})$. Furthermore, for any  $\mathbb F$-locally bounded optional process $H$ and any $K\in{\mathcal{I}}^o_{\loc}(N^{\mathbb G},{\mathbb G})$, we have $KH\in {\mathcal{I}}^o_{\loc}(N^{\mathbb G},{\mathbb G})$. As a result, for any $\mathbb F$-locally bounded local martingale $M$, we have
 $[M,K\is N^{\mathbb G}]=(\Delta M)K\is N^{\mathbb G}$ is a $\mathbb G$-local martingale since
  $(\Delta M)K\in{\mathcal{I}}^o_{\loc}(N^{\mathbb G},{\mathbb G})$. This proves assertion (c), and ends the proof of the theorem.
  \end{proof}

The rest of this subsection introduces the pure default (local) martingales {\it of the second type}.
 After posting our first version of the paper on Arxiv, and presenting it in several conferences, some colleagues informed us about the
  existence of this class of martingales in \cite{azemaetal93} for the Brownian framework and honest times (only) that avoids $\mathbb F$-stopping times. To introduce our class that extends \cite{azemaetal93} to the general framework, 
  we start with the following notation. On the set 
  $\left(\Omega\times[0,+\infty), {\cal F}\otimes{\cal B}(\mathbb R^+)\right)$ (where ${\cal B}(\mathbb R^+)$
  is the Borel $\sigma$-field on $\mathbb R^+=[0,+\infty)$), we consider 
$$\mu(d\omega,dt):=P(d\omega)dD_t(\omega), $$ which is a finite measure and hence it can be normalized into a probability measure. Recall that the predictable,
optional, and progressive sub-$\sigma$-fields are denoted by ${\cal P}(\mathbb F)$, ${\cal O}(\mathbb F)$,
and Prog$(\mathbb F)$ respectively. On $(\Omega, {\cal F})$, we consider the sub-$\sigma$-fields
${\cal F}_{\tau-}$, ${\cal F}_{\tau}$, and ${\cal F}_{\tau+}$ obtained as the sigma fields generated by
$\{X_{\tau}\ \big|\ X\ \mbox{is $\mathbb F$-predictable}\}$, $\{X_{\tau}\ \big|\ X\ \mbox{is $\mathbb F$-optional}\}$,
and $\{X_{\tau}\ \big|\ X\ \mbox{is $\mathbb F$-progressively measurable}\}$ respectively.
Furthermore, for any ${\cal H}\in\{ {\cal P}(\mathbb F), {\cal O}(\mathbb F), \mbox {Prog}(\mathbb F)\}$, for any $p\in [1,+\infty)$, we define
\begin{equation}\label{L1(PandD}
L^p\left({\cal H}, P\otimes D\right):=\Bigl\{ X\ {\cal H}\mbox{-measurable}\
\big|\ \ \E[\vert X_{\tau}\vert^p I_{\{\tau<+\infty\}}]=:\E_{P\otimes D}[\vert X\vert^p]<+\infty\Bigr\},\end{equation}
and its localisation
\begin{eqnarray*}\label{L1(PandD)Local}
L^p_{loc}\left({\cal H}, P\otimes D\right):=\left\{ X\ \big|\ X^{T_n}\in L^p\left({\cal H},
P\otimes D\right),\ T_n\ \mbox{$\mathbb F$-stopping time s.t.}\ \sup_n T_n=+\infty\right\}.\end{eqnarray*}

\noindent In the following, we define the  pure default ( local) martingales {\it of the second type}, and specify their relationship to the first type
 of pure default martingales.

\begin{theorem}\label{the3spaces}
The following assertions hold.\\
(a) The class of processes
\begin{align}
{\cal M}^{(2)}_{loc}(\mathbb G)&:= \left\{ k\is D\ \Big|\ k\in L^1_{loc}\left({\rm{Prog}}(\mathbb F), P\otimes D\right)\ \
\mbox{and}\ \ \  \E[k_{\tau}|{\cal F}_{\tau}]I_{\{\tau<+\infty\}}=0\ P\text{-a.s}\  \right\}\nonumber
\end{align}
is a space of pure default local martingales, that we call the class of pure default local martingales of the second type.\\
(b) For any $k\in L^1_{loc}\left(\mbox {Prog}(\mathbb F), P\otimes D\right)$ and any
 $h\in {\cal I}^{o}_{loc}(N^{\mathbb G},\mathbb G)$
we have
\begin{eqnarray*}\label{orthogonality}
[k\is D,h\is N^{\mathbb G}]\in {\cal A}_{loc}(\mathbb G) \ \mbox{if and only if}\ \
[k\is D,h\is N^{\mathbb G}]\in{\cal M}_{loc}(\mathbb G)
\end{eqnarray*}
(i.e. the first type of pure default local martingales are orthogonal to the second type of pure default local martingales provided the local integrability of their product). 
\end{theorem}

\begin{proof} It is clear that ${\cal M}^{(2)}_{loc}(\mathbb G)$ is a subspace of ${\cal M}_{loc}(\mathbb G)$ on the one hand. On
the other hand, 
$$[k\is D, M]=(\Delta M)k\is D\in {\cal M}^{(2)}_{loc}(\mathbb G),$$
for any $k\in L^1_{loc}\left(\mbox{Prog}(\mathbb F), P\otimes D\right)$ satisfying
 $\E(k_{\tau}\ \big|\ {\cal F}_{\tau})=0$ $P$-a.s on $\{\tau<+\infty\}$, and for any bounded $\mathbb F$-local martingale $M$. This proves assertion (a), and
  the remaining part of this proof focuses on proving assertion (b).
 To this end,  let $k\in L^1_{loc}\left(\mbox {Prog}(\mathbb F), P\otimes D\right)$
such that $\E[k_{\tau}|{\cal F}_{\tau}]=0\ P$-a.s on $\{\tau<+\infty\}$, and $K\in {\mathcal{I}}^o_{\loc}(N^{\mathbb G},\mathbb G)$. Then, we have
$$
[K\is N^{\mathbb G}, k\is D]= kK\Delta N^{\mathbb G}\is D=k{\widetilde K}\is D,\ \ \
 \ \  \ \widetilde K:=KG(\widetilde G)^{-1}I_{\{\widetilde G>0\}}.$$
 Thus, $[K\is N^{\mathbb G}, k\is D]\in {\cal A}_{loc}(\mathbb G)$ if and only if
 $k{\widetilde K}\in L^1_{loc}\left(\mbox {Prog}(\mathbb F), P\otimes D\right)$, and in this case we have 
 $$\E(k_{\tau}\widetilde{K}_{\tau}\ \big|\ {\cal F}_{\tau})I_{\{\tau<+\infty\}}=
 \widetilde{K}_{\tau}\E(k_{\tau}\ \big|\ {\cal F}_{\tau})I_{\{\tau<+\infty\}}=0,\ P\mbox{-a.s.}.$$ This ends the proof of the theorem.\end{proof}

\begin{proposition}\label{IntegOptionalProcess} Let $H$ be an $\mathbb F$-optional process.
Then the following hold.\\
{\rm{(a)}}  If both $K$ and $HK$ belong to ${\mathcal{I}}^o_{\loc}(N^{\mathbb G},\mathbb G)$, then $ HK\is N^{\mathbb G}= H\is(K\is N^{\mathbb G})$ is  a $\mathbb G$-local martingale. In particular,
$\left(K\is N^{\mathbb G}\right)^{\sigma-}\in {\cal M}_{loc}(\mathbb G)$, for any $\mathbb F$-stopping time $\sigma$, and any $K\in{\mathcal{I}}^o_{\loc}(N^{\mathbb G},\mathbb G)$.\\
{\rm{(b)}} If both $k$ and $kH$ belong to $L^1_{loc}\left({\rm {Prog}}(\mathbb F), P\otimes D\right)$ 
and $\E(k_{\tau}\ \big|\ {\cal F}_{\tau})I_{\{\tau<+\infty\}}=0,\ P\mbox{-a.s.}$, then
$ Hk\is D=H\is(k\is D)$ belongs to ${\cal M}^{(2)}_{loc}(\mathbb G)$. In particular,
$\left(k\is D\right)^{\sigma-}\in {\cal M}^{(2)}_{loc}(\mathbb G)$, for any $\mathbb F$-stopping time $\sigma$, and any $k\in L^1_{loc}\left({\rm {Prog}}(\mathbb F), P\otimes D\right)$ such that 
$\E(k_{\tau}\ \big|\ {\cal F}_{\tau})I_{\{\tau<+\infty\}}=0,\ P\mbox{-a.s.}$.
\end{proposition}

The proof of this proposition is obvious and will be omitted.\\

In \cite{jeulin80}, the author considers
 \begin{equation}\label{Jeulinspace}
 {\cal M}^{(3)}_{loc}(\mathbb G):=\{k\is D\ \big|\ \ k\in L^1_{loc}\left(\mbox {Prog}(\mathbb F), P\otimes D\right)\ \&\
 \ \E\left(k_{\tau}\ \big| {\cal F}_{\tau-}\right)I_{\{\tau<+\infty\}}=0\ P\mbox{-a.s.}\}.\end{equation}
This class contains ${\cal M}^{(2)}_{loc}(\mathbb G)$ of Theorem \ref{the3spaces}, while in general it fails to satisfy Proposition \ref{IntegOptionalProcess}. Hence, the elements of ${\cal M}^{(3)}_{loc}(\mathbb G)$ (Jeulin's space)  cannot be pure default martingales, and are not orthogonal to the pure default martingales of the first type.\\

 \noindent Given that we singled out two orthogonal types of pure default local martingales, one can ask naturally the following.
\begin{equation}\label{QPMM2}
 \mbox{How many types of orthogonal pure default martingales are there?}\end{equation}
The answer to this difficult question, as well as to all the unanswered previous questions, boils down to the complete decomposition of any
$\mathbb G$-martingale stopped at $\tau$ into the sum of orthogonal (local) martingales. This is the aim of the following subsection.

%%%%%%%%%%%%%%%%%%%%%%%%%%%%%%%%%%%%%%%%%%%%%%%%%%%%%%
%%%%%%%%
%%%%%%%%%
%%%%%%%%%%%%%%%%%%%%%%%%%%%%%%%%%%%%%%%%%%%%%%%%%%%%%%%%
\subsection{The optional martingale representation theorems}\label{subsection2.2}

 This subsection elaborates our complete, rigorous, explicit and general optional representation theorem for any
 $\mathbb G$-martingale stopped at $\tau$. To this end, we start by decomposing a class of $\mathbb G$-martingales that is widely used in insurance
 (mortality/longevity) derivatives and credit risk derivatives. These martingales take the form of
  $(E[h_{\tau}\ \big|\ {\cal G}_t],\ t\geq 0),$ where $h$ represents the payoff process with adequate integrability and measurability condition(s).
  To state our optional martingale representation for these martingales, we recall an interesting result
  of \cite{aksamitchoullijeanblanc15}, and we give a technical lemma afterwards.

\begin{theorem} \cite[Theorem 3]{aksamitchoullijeanblanc15} For any $\mathbb F$-local martingale $M$, the following
\begin{equation} \label{processMhat}
\widehat{M} := M^\tau -{\widetilde{G}}^{-1} I_{\Lbrack 0,\tau\Rbrack} \is [M,m] +
 I_{\Lbrack 0,\tau\Rbrack} \is\Big(\Delta M_{\widetilde R} I_{\Lbrack\widetilde R,+\infty\Lbrack}\Big)^{p,\mathbb F},
\end{equation}
is a $\mathbb G$-local martingale. Here
\begin{equation} \label{timetildeR}
R:= \inf \{t\ge 0: G_t =0 \},\quad\mbox{and}\quad \widetilde R := R_{\{\widetilde{G}_R=0<G_{R-}\}} :=
R I_{\{\widetilde{G}_R=0<G_{R-}\}} + \infty I_{\Omega\setminus\{\widetilde{G}_R=0<G_{R-}\}}.
\end{equation}
\end{theorem}

\begin{remark}\label{orthogonalityofMhat}
 Due to $\Delta m={\widetilde G}-G_{-}$, on $\Rbrack0,\tau\Rbrack$,  $\Delta \widehat{M}$ coincides with the
$\mathbb F$-optional process 
$$\widetilde K := (\Delta M)G_{-}{\widetilde{G}}^{-1}I_{\{\widetilde G>0\}}+
\ ^{p,\mathbb F}\Big(\Delta M_{\widetilde R}I_{\Lbrack\widetilde R\Rbrack}\Bigr).$$
Thus, we conclude that $\widehat M$ is orthogonal to both types (first type and second type) of pure default local martingales defined in the previous subsection. This fact follows directly from Proposition \ref{IntegOptionalProcess}.
\end{remark}

\begin{lemma}\label{LemmaTech} Let $h\in L^1({\cal O}(\mathbb F),P\otimes D)$. Then both $h$ and
$\left(M^h - h \is D^{o,\mathbb F}\right)G^{-1}I_{\Lbrack 0,R\Lbrack}$ belong to ${\cal I}^o_{\loc}(N^{\mathbb G},\mathbb G)$, where
\begin{equation} \label{processesMhandJ}
M^h_t := \ ^{o,\mathbb F}\Big(\int_0^\infty h_u  dD_u^{o,\mathbb F}\Big)_t=\E\Bigl[\int_0^\infty h_u  dD_u^{o,\mathbb F}\ \big|\ {\cal F}_t\Bigr].
	\end{equation}
\end{lemma}

 The proof of this lemma is relegated to  {Appendix} \ref{SubsectionproofTheorem1}. The following is one of the principal results about our martingale representation. 
\begin{theorem} \label{TheoRepresentation} Let $h\in L^1({\cal O}(\mathbb F),P\otimes D)$, and $M^h$ be given in (\ref{processesMhandJ}).
 Then the following hold.\\
	{\rm{(a)}} The $\mathbb G$-martingale $H_t := \!\! \ ^{o,\mathbb G}(h_\tau)_t=\E[h_\tau | {\cal G}_t]$ admits the following representation.
	\begin{equation}\label{RepresentationofH}
	H - H_0 ={{I_{\Rbrack 0,\tau\Rbrack}}\over{G_-}} \is \widehat{M^h}
-{{M^h_{-} - (h \is D^{o,\mathbb F})_{-}}\over{G_{-}^2}} I_{\Rbrack 0,\tau\Rbrack} \is \widehat{m} +
{{hG-M^h +h \is D^{o,\mathbb F}}\over{G}}I_{\Rbrack 0,R\Lbrack} \is N^{\mathbb G}.
	\end{equation}
	{\rm{(b)}} If $h\in L\log L({\cal O}(\mathbb F),P\otimes D)$
(i.e.\ $\E[\vert h_{\tau}\vert\ln(\vert h_{\tau}\vert)I_{\{\tau<+\infty\}}]=
\E[\int_0^{\infty}\vert h_u\vert \log(\vert h_u\vert)dD_u ]<+\infty$), then both
$\left(hG-M^h +h \is D^{o,\mathbb F}\right)G^{-1}I_{\Rbrack 0,R\Lbrack} \is N^{\mathbb G}$
 and $ G_-^{-1}I_{\Rbrack 0,\tau\Rbrack} \is  \widehat{M^h} -
 \left(M^h_{-} - (h \is D^{o,\mathbb F})_{-}\right)G_{-}^{-2} I_{\Rbrack 0,\tau\Rbrack} \is \widehat{m}$
  are uniformly integrable \(\mathbb G\)-martingales.\\
	{\rm{(c)}} If $h\in L^2({\cal O}(\mathbb F),P\otimes D)$ , then the two \(\mathbb G\)-martingales
 $\left(hG-M^h +h \is D^{o,\mathbb F}\right)G^{-1}I_{\Rbrack 0,R\Lbrack} \is N^{\mathbb G}$ and
 $ G_-^{-1}I_{\Rbrack 0,\tau\Rbrack} \is  \widehat{M^h}
 -\left(M^h_{-} - (h \is D^{o,\mathbb F})_{-}\right)G_{-}^{-2}I_{\Rbrack 0,\tau\Rbrack} \is \widehat{m}$
 are square integrable  and orthogonal martingales. 
\end{theorem}

\noindent For the sake of easy exposition, we postpone the proof of the theorem to {Appendix}  \ref{SubsectionproofTheorem1}.
Theorem \ref{TheoRepresentation} states that the risk with terminal value $h_{\tau}$ for some $h\in L^1({\cal O}(\mathbb F),P\otimes D)$,
can be decomposed into three orthogonal risks: The ``pure" financial risk which is the first term in the RHS of (\ref{RepresentationofH}),
while the second term of the RHS represents the resulting risk from correlation between the market model and $\tau$. The last term in the RHS of
(\ref{RepresentationofH}) models the pure default risk of type one.\\

\noindent Below, we illustrate our  theorem on particular models for the pair $(\tau, \mathbb F)$.

\begin{corollary} \label{CorRepresentation} Let $h\in L^1({\cal O}(\mathbb F),P\otimes D)$,  and consider $\overline{N}^{\mathbb G}$, $N^{\mathbb G}$ and $M^h$ given by
 (\ref{processNGbar}), (\ref{processNG})  and  (\ref{processesMhandJ}) respectively.
Then the optional representation (\ref{RepresentationofH}) takes the following forms.\\	
{\rm {(a)}} If \(\tau\) is an \(\mathbb F\)-pseudo stopping time (i.e. $\E[M_{\tau}]=\E[M_0]$ for any bounded $\mathbb F$-martingale $M$),
 then
\[
H - H_0 = G_-^{-1}I_{\Rbrack 0,\tau\Rbrack} \is  M^h
+\Bigl[Gh-M^h+h \is D^{o,\mathbb F}\Bigr]G^{-1}I_{\Rbrack 0,R\Lbrack} \is N^{\mathbb  G}  .
\]
{\rm {(b)}} If \(\tau\) avoids all \(\mathbb F\)-stopping times, then
\begin{equation}\label{assersionb}
H - H_0 ={{I_{\Rbrack 0,\tau\Rbrack}}\over{ G_-}}\is\widehat{M^h}
-{{M^h_{-}-(h \is D^{o,\mathbb F})_{-}}\over{G_-^{2}}} I_{\Rbrack 0,\tau\Rbrack}  \is\widehat{m}
 + {{G_{-}\ ^{p,\mathbb F}(h)-M^h_{-}+\ ^{p,\mathbb F}(h) \is D^{p,\mathbb F}}\over{G_{-}}}I_{\{G_{-}>0\}}\is
  \overline{N}^{\mathbb G}  .
\end{equation}
{\rm {(c)}} If  all \(\mathbb F\)-martingales are continuous, then it holds that
\begin{equation}\label{assersionc}
H - H_0 = G_-^{-1}I_{\Rbrack 0,\tau\Rbrack}  \is  \overbar{M^h}
-{{M^h_{-}-(h \is D^{o,\mathbb F})_{-}}\over{G_-^{2}}}I_{\Rbrack 0,\tau\Rbrack}  \is\overline{m}
 + {{\ ^{p,\mathbb F}(h)G-M^h_{-}+\ ^{p,\mathbb F}(h) \is D^{p,\mathbb F}}\over{G}}I_{\Rbrack 0,R\Lbrack} \is\overline{N}^{\mathbb G}.
\end{equation}
Here, for any $\mathbb F$-local martingale $M$, $\overline{M}$ is defined by
\begin{equation}\label{Mbarprocess}
\overline{M} :=
 M^\tau - {G}^{-1}_{-}I_{\Rbrack 0,\tau\Rbrack}  \is\langle M,m\rangle^{\mathbb F}.
\end{equation}

\end{corollary}

\begin{proof}
1) Thanks to \cite[Theorem 1]{nikeghbaliyor05}, it holds that \(\tau\) is an \(\mathbb F\)-pseudo stopping time if and only if  \(m\equiv m_0\).
 This leads to $\widehat{m}\equiv m_0$ and \(\widetilde{G}=G_-\). Therefore, we get  \(\{\widetilde{G}=0<G_{-}\}=\Lbrack \widetilde R\Rbrack=\emptyset\),
\begin{eqnarray*}\label{MhatMbar}
\widehat M\equiv M^\tau\quad\mbox{for any}\quad M\in {\cal M}_{\loc}(\mathbb F),\quad\mbox{and}\quad I_{\Rbrack 0,\tau\Rbrack}  \is\widehat{m}\equiv 0.\end{eqnarray*}
Thus, the proof of assertion (a) follows from combining the latter fact with Theorem \ref{TheoRepresentation}. \\
2) Suppose that $\tau$ avoids $\mathbb F$-stopping times. Then $\tau < R$ $P$-a.s. (since $\tau\leq R\ P$-a.s.) and it is easy to check that
 \(\widetilde{G}=G\), $D^{o,\mathbb F}\equiv D^{p,\mathbb F}$ is continuous,
  $I_{\Lbrack 0,R\Lbrack} \is D\equiv I_{\Lbrack 0,R\Rbrack}\is D=D$, $N^{\mathbb G}={\overline{N}}^{\mathbb G}$, and
\[
\{\ ^{p,\mathbb F}h\not=h\ \mbox{or}\ G\not=G_{-}\ \mbox{or}\ M^h\not=M^h_{-}\} \cap \Lbrack \tau \Rbrack =\emptyset .
\]
Therefore, assertion (b) follows from combining all these remarks with Theorem \ref{TheoRepresentation}.\\
3) Suppose that all \(\mathbb F\)-martingales are continuous. Then $M^h$ and $m$ are continuous, and
\begin{eqnarray*}
 \widehat {M^h}=\overline{M^h},\qquad  \widehat m=\overbar m,\qquad \widetilde G=G_{-},\qquad  D^{o,\mathbb F}=
 D^{p,\mathbb F},\qquad \mbox{and}\qquad N^{\mathbb G}={\overline{N}}^{\mathbb G}.\end{eqnarray*}
 Furthermore, all $\mathbb F$-stopping times are predictable. As a result, $R$ is predictable and
  $G_{R-}=0$ on $\{R<+\infty\}$. This implies that $\Lbrack 0,\tau\Rbrack\subset \Lbrack 0,R\Lbrack$.
   Therefore, a combination of these remarks with
\begin{eqnarray*}
{{hG-M^h+h \is D^{o,\mathbb F}}\over{G}}I_{\Rbrack 0,\tau\Rbrack}  &=&
{{hG_{-}-M^h_{-}+(h \is D^{o,\mathbb F})_{-}}\over{G_{-}}}I_{\Rbrack 0,\tau\Rbrack}
- \Delta{{M^h-h \is D^{o,\mathbb F}}\over{G}}I_{\Rbrack 0,\tau\Rbrack}\\
&  =&
  {{hG_{-}-M^h_{-}+(h \is D^{o,\mathbb F})_{-}}\over{G_{-}}}I_{\Rbrack 0,\tau\Rbrack}  -
   {{hG_{-}-M^h_{-}+(h \is D^{o,\mathbb F})_{-}}\over{GG_{-}}} \Delta G I_{\Rbrack 0,\tau\Rbrack}  \\
& =& {{hG_{-}-M^h_{-}+(h \is D^{o,\mathbb F})_{-}}\over{G}} I_{\Rbrack 0,\tau\Rbrack},
\end{eqnarray*}
proves assertion (c) . This ends the proof of the corollary.
 \end{proof}

\noindent It is worth mentioning that the pseudo-stopping time model for $\tau$ covers several cases proposed and/or  treated in the literature.  In fact, the pseudo-stopping time model covers the immersion case (i.e. $M$ is a $\mathbb G$-local martingale for any $\mathbb F$-local martingale $M$) or equivalently the Cox' time given by $\tau:=\inf\{t\geq 0\ \big|\ S_t\geq E\}$ where $E$ being a random variable that is independent of ${\cal F}_{\infty}$, the case when $\tau$ is independent of ${\cal F}_{\infty}$
 (no correlation between the financial
   market and the random time), and the case when $\tau$ is an $\mathbb F$-stopping time (i.e.\  the case when the random time is fully observable through the public flow of information $\mathbb F$). For more details about pseudo-stopping times, we refer the reader to \cite{nikeghbali06,nikeghbaliyor05}.\\

%\begin{remark} \label{RemarkSpecialCases}
 \noindent Corollary \ref{CorRepresentation} tells us that our representation (\ref{RepresentationofH})
 goes beyond the context of \cite{blanchetjeanblanc04}.
 Indeed, our assertions (b) and (c) above extend \cite{blanchetjeanblanc04} to the case where $h$ is $\mathbb F$-optional (the condition (\ref{assumption2}) is relaxed), 
 as is the case for some examples in \cite{gerbershiu13}, and $G$ might vanish on the one hand. On the other hand, by comparing the RHS terms of
 (\ref{assersionb}) and (\ref{assersionc}), we deduce that their third type of risk
 (the integrands with respect to $\overline{N}^{\mathbb G}$) differ tremendously,
 and they cannot be written in a universal form using ${\overline N}^{\mathbb G}$. This explains
 why the representation of \cite{blanchetjeanblanc04} might fail for general $\mathbb F$-optional $h$. To see how our results extend this latter paper,
  we consider $h\in L^1({\cal P}(\mathbb F),P\otimes D)$ and put
\begin{eqnarray*} \label{processesmhandJ}
m^h := \ ^{o,\mathbb F}\Big(\int_0^\infty h_u  dF_u \Big),\ \ \mbox{where}\ \ \ F:=1-G.\end{eqnarray*}
Then it is not difficult to deduce that $\widehat{M^h}$ defined in (\ref{processesMhandJ}) and $\widehat{m^h}$ are related by  $\widehat{M^h} = \widehat{m^h} + h \is \widehat{m}$. As a result, the decomposition (\ref{RepresentationofH}) takes the form of
$$
H - H_0 ={{I_{\Rbrack 0,\tau\Rbrack}}\over{G_-}} \is  \widehat{m^h}
+{{G_{-}h-m^h_{-} +(h \is F)_{-}}\over{G_-^{2}}} I_{\Rbrack 0,\tau\Rbrack} \is\widehat{m}
+ {{Gh-m^h+ h \is F }\over{G}}I_{\Rbrack 0,R\Lbrack} \is N^{\mathbb G}  .
$$
In particular when either $\tau$ avoids $\mathbb F$-stopping times or all $\mathbb F$-local martingales are continuous,  the representation (\ref{RepresentationofH}) becomes 
$$
H - H_0 = {{I_{\Rbrack 0,\tau\Rbrack}}\over{G_{-}}} \is\widehat{m^h}
+{{hG_{-}-m^h_{-} +(h \is F)_{-}}\over{G_-^2}} I_{\Rbrack 0,\tau\Rbrack} \is\widehat{m}
+{{Gh-m^h+h \is F}\over{G}}{ \is\overline{N}^{\mathbb G} .}
$$
Even this latter formula extends \cite{blanchetjeanblanc04} to the case where $G$ might vanish, as the condition $G>0$ is important in the proofs of this paper.\\

The  rest of this subsection focuses on decomposing an arbitrary $\mathbb G$-martingale stopped at $\tau$.
 To this end, we need the following intermediate result that is simple but important. It explains how the general case can
  always be reduced to the case of Theorem \ref{TheoRepresentation}.

\begin{proposition}\label{ConditionalExpectation}
The following assertions hold.\\
{\rm{(a)}} Let $X$ be a measurable process such that $X\geq 0\ \mu$-a.e.
 (recall that $\mu:=P\otimes D$) or $X$ belongs to $L^1({\cal F}\otimes{\cal B}(\mathbb R^+),P\otimes D)$. Then the following equalities hold $P$-a.s. on $\{\tau<+\infty\}$.
$$
\E_{\mu}\left[X\big|{\cal P}(\mathbb F)\right](\tau)=\E\left[X_{\tau}\big|{\cal F}_{\tau-}\right],
\ \ \E_{\mu}\left[X\big|{\cal O}(\mathbb F)\right](\tau)=\E\left[X_{\tau}\big|{\cal F}_{\tau}\right],
\ \E_{\mu}\left[X\big|{\rm{Prog}}(\mathbb F)\right](\tau)=\E\left[X_{\tau}\big|{\cal F}_{\tau+}\right]. $$
Here $\E_{\mu}[.|.]$ is the conditional expectation under the finite measure $\mu$.\\
{\rm{(b)}} If $k\in L^1\left({\rm{Prog}}(\mathbb F), P\otimes D\right)$, then 
there exists a unique ($\mu$-a.e.) $h\in  L^1\left({\cal O}(\mathbb F), P\otimes D\right)$ satisfying
\begin{eqnarray*}\label{Opt/Progressive}
\E\left[k_{\tau}\ \big|\ {\cal F}_{\tau}\right]=h_{\tau}\ \ \ \ P\text{-a.s.}\ \ \ \mbox{on}\ \ \{\tau<+\infty\}.\end{eqnarray*}
\end{proposition}

\begin{proof} The proof of assertion (a) is obvious and will be omitted. Assertion (b) follows immediately from assertion (a)
by putting $h=\E_{\mu}[k|{\cal O}(\mathbb F)]$, and the proof of the proposition is completed.
\end{proof}

\noindent The following states our full optional martingale representation result.

\begin{theorem}\label{TheoRepresentation2}
For any $\mathbb G$-martingale, $M^{\mathbb G}$, there exist two processes $h\in L^1\left({\cal O}(\mathbb F),P\otimes D\right)$
 and $k\in L^1\left({\rm{Prog}}(\mathbb F),P\otimes D\right)$, such that $\E[k_{\tau}|{\cal F}_{\tau}]=0$,
 $k_{\tau}+h_{\tau}=M^{\mathbb G}_{\tau}$ $P$-a.s.\ on
  $\{\tau<+\infty\}$, and
\begin{equation}\label{Representation22}
\left( M^{\mathbb G}\right)^{\tau}-M^{\mathbb G}_0={{I_{\Rbrack 0,\tau\Rbrack}}\over{G_-}}  \is\widehat{M^h}
- {{M^h_{-}-(h \is D^{o,\mathbb F})_{-}}\over{G_-^{2}}} I_{\Rbrack 0,\tau\Rbrack} \is \widehat{m}
+ {{hG-M^h+h \is D^{o,\mathbb F}}\over{G}}I_{\Rbrack 0,R\Lbrack}\is N^{\mathbb  G} +k\is D.
\end{equation}
Here $M^h$ and $m$ are defined in (\ref{processesMhandJ}) and (\ref{GGtildem}) respectively.
\end{theorem}

\begin{proof} Let $M^{\mathbb G}$ be a $\mathbb G$-martingale. Then, on the one hand, there exists (unique up to
 $P\otimes D$-a.e.) $k^{(1)}\in L^1\left({\rm{Prog}}(\mathbb F),P\otimes D\right)$ such that $M^{\mathbb G}_{\tau}=k^{(1)}_{\tau}$ \(P\)-a.s. on $\{\tau<+\infty\}$, see  \cite[Lemma B.1]{aksamit2017} for details about this fact. Thanks to Proposition \ref{ConditionalExpectation}--(b), there exists $h\in L^1\left({\cal O}(\mathbb F),P\otimes D\right)$ such that $\E[k^{(1)}_{\tau}|{\cal F}_{\tau}]=h_{\tau}$ \(P\)-a.s. on $\{\tau<+\infty\}$. On the other hand, remark that ${\cal G}_t\cap\left(\tau>t\right)\subset{\cal F}_{\tau}$ and put $k:=k^{(1)}-h$. Then we conclude that
$k\in L^1\left({\rm{Prog}}(\mathbb F),P\otimes D\right)$ and satisfies $\E[k_{\tau}|{\cal F}_{\tau}]=0$ \(P\)-a.s. on $\{\tau<+\infty\}$. This implies that $E(k_{\tau}\big|{\cal G}_t)I_{\{\tau>t\}}= E(E(k_{\tau}\big|{\cal F}_{\tau})\big|{\cal G}_t)I_{\{\tau>t\}}=0$, and 
\begin{eqnarray*}
	M^{\mathbb G}_{t\wedge\tau}&=&\E\big[M^{\mathbb G}_{\tau}\ \big|\ {\cal G}_t\big]=\E\big[k^{(1)}_{\tau}|{\cal G}_t\big]=\E\big[k_{\tau}|{\cal G}_t\big]+\E\big[h_{\tau}|{\cal G}_t\big]\\
&=&k_{\tau}I_{\Lbrack \tau,+\infty\Lbrack}(t)+\E\left[k_{\tau}I_{\{\tau>t\}}\big|{\cal G}_t\right]+\E\big[h_{\tau}|{\cal G}_t\big]=k\is D_t+\E\left[h_{\tau}\big|{\cal G}_t\right].\end{eqnarray*}
Hence, a direct application of Theorem \ref{TheoRepresentation} to $\E\left[h_{\tau}\big|{\cal G}_t\right]$,
 the decomposition (\ref{Representation22}) follows immediately, and the proof of theorem is completed.\end{proof}

Theorem \ref{TheoRepresentation2} can be slightly reformulated as follows.

\begin{theorem}\label{theo4MartingaleDecomposGeneral} Let  $M^{\mathbb G}$ be a $\mathbb G$-martingale and $R$ be given in (\ref{timetildeR}).
 Then there exists a unique triplet $(M^{\mathbb F},\varphi^{(o)}, \varphi^{(pr)}) $ that belongs to $ {\cal M}_{0,loc}(\mathbb F)\times {\cal I}^o_{loc}\left(N^{\mathbb G},\mathbb G\right)\times  L^1_{loc}\left(\widetilde\Omega, {\rm{Prog}}(\mathbb F), P\otimes D\right)$  and satisfies  
 \begin{eqnarray}\label{Condition1}
M^{\mathbb F}=(M^{\mathbb F})^R,\quad  \Delta M^{\mathbb F}I_{\{\widetilde G=0\}}=0,\quad  \varphi^{(o)}=\varphi^{(o)}I_{\Lbrack 0,R\Lbrack},\quad \E\left[\varphi^{(pr)}_{\tau}\ \big|\ {\cal F}_{\tau}\right]I_{\{\tau<+\infty\}}=0,\quad  P\mbox{-a.s.},
\end{eqnarray}
and
\begin{equation}\label{MartingaleDecomposGeneral}
\left(M^{\mathbb G}\right)^{\tau}=M^{\mathbb G}_0+G_{-}^{-2}I_{\Rbrack 0,\tau\Rbrack}\is\widehat{ M^{\mathbb F}}
+\varphi^{(o)}\is N^{\mathbb G}+\varphi^{(pr)}\is D.\end{equation}
\end{theorem}

\begin{proof} It is clear that the existence of the triplet $\left(M^{\mathbb F}, \varphi^{(o)}, \varphi^{(pr)}\right)$, for which
the decomposition (\ref{MartingaleDecomposGeneral}) holds, follows immediately from Theorem \ref{TheoRepresentation2}
  by putting $M^{\mathbb F}:=G_{-}\is M^h-\left(M^h_{-}-(h \is D^{o,\mathbb F})_{-}\right)\is m$. Therefore, the proof of the existence follows immediately as soon as we prove that this $M^{\mathbb F}$ satisfies the first two equalities of (\ref{Condition1}).  To this end, we remark that $I_{\Rbrack R,+\infty\Lbrack}\is G=I_{\Rbrack R,+\infty\Lbrack}\is D^{o,\mathbb F}=0$, due to $\E[I_{\Rbrack R,+\infty\Lbrack}\is D^{o,\mathbb F}_{\infty}]= P(R<\tau<+\infty)=0$. This implies that $I_{\Rbrack R,+\infty\Lbrack}\is m=0$ on the one hand. On the other hand, thanks to (\ref{processesMhandJ}), it is clear that for any $\mathbb F$-stopping time $\sigma$
  \begin{equation}\label{MhConstant}
   M^h_{\sigma}=(h\is D^{o,\mathbb F})_{\sigma}+E(h_{\tau}I_{\{\sigma<\tau<+\infty\}}\big|\ {\cal F}_{\sigma})=(h\is D^{o,\mathbb F})_R\quad P\mbox{-a.s.}\quad \mbox{on}\quad\{\sigma\geq R\}.\end{equation}
   As a result, we also get  $I_{\Rbrack R,+\infty\Lbrack}\is M^h=0$, and deduce that $M^{\mathbb F}=(M^{\mathbb F})^R$. Furthermore, we calculate 
\begin{eqnarray*}
  \Delta M^{\mathbb F}I_{\{\widetilde G=0\}}&=&\Bigl[G_{-}\Delta M^h-(M^h_{-}-(h \is D^{o,\mathbb F})_{-})\Delta m\Bigr]I_{\{\widetilde G=0\}}\\
  &=&G_{-}\Bigl[M^h-(h \is D^{o,\mathbb F})_{-}\Bigr]I_{\{\widetilde G=0\}}= G_{-}h (\widetilde G-G)I_{\{\widetilde G=0\}}=0.\end{eqnarray*}
  
  The second equality is due to $\Delta m=\widetilde G-G_{-}$, while the third equality follows from combining (\ref{MhConstant}) and $\{\widetilde G=0\}\subset \{G=0\}=\Lbrack R,+\infty\Lbrack$. This proves the existence of $\left(M^{\mathbb F}, \varphi^{(o)}, \varphi^{(pr)}\right)$ that satisfies both (\ref{Condition1}) and (\ref {MartingaleDecomposGeneral}). Thus, the remaining part of this proof focuses on the uniqueness of this triplet. To this end, we suppose the existence of such triplet  $\left(M^{\mathbb F}, \varphi^{(o)}, \varphi^{(pr)}\right)$ satisfying (\ref{Condition1}) and 
\begin{equation}\label{UniqunessEquation}
0=G_{-}^{-2}I_{\Rbrack 0,\tau\Rbrack}\is\widehat{ M^{\mathbb F}}
+\varphi^{(o)}\is N^{\mathbb G}+\varphi^{(pr)}\is D.\end{equation}
For $n\geq 1$, we put
$$\Gamma_n:=\left\{{\widetilde G}^{-1}+\vert\Delta M^{\mathbb F}\vert+\vert \varphi^{(o)}\vert
+\vert \Delta\left(\Delta M^{\mathbb F}_{\widetilde R}I_{\Lbrack\widetilde{R},+\infty\Lbrack}\right)^{p,\mathbb F}\vert\leq n\right\},$$
and by utilizing (\ref{UniqunessEquation}), we conclude that 
\begin{eqnarray*}
I_{\Gamma_n}\is [\varphi^{(pr)}\is D,\varphi^{(pr)}\is D]=
-I_{\Gamma_n}\Bigl({{\Delta\widehat{ M^{\mathbb F}}}\over{G_{-}^2}}+\varphi^{(o)}\Delta N^{\mathbb G}\Bigr)\varphi^{(pr)}\is D
= -I_{\Gamma_n}\Bigl({{\Delta M^{\mathbb F}}\over{G_{-}\widetilde G}}+
{{G}\over{\widetilde G}}\varphi^{(o)}\Bigr)\varphi^{(pr)}\is D\end{eqnarray*}
is a local $\mathbb G$-martingale.Thus, $I_{\Gamma_n}\is [\varphi^{(pr)}\is D,\varphi^{(pr)}\is D]$
 is a null process since it is nondecreasing. By combining this with the fact that $\Gamma_n\cap \Lbrack0,\tau\Rbrack$
increases to $\Lbrack0,\tau\Rbrack$,  we deduce that $[\varphi^{(pr)}\is D,\varphi^{(pr)}\is D]\equiv 0$
 or equivalently $\varphi^{(pr)}\equiv 0$ $P\otimes D$-a.e.. 
Similarly, we derive
\begin{equation*}
I_{\Gamma_n}\is [\varphi^{(o)}\is N^{\mathbb G}, \varphi^{(o)}\is N^{\mathbb G}] =
-I_{\Gamma_n}\Bigl({{\Delta\widehat{ M^{\mathbb F}}}\over{G_{-}^2}}\Bigr)\varphi^{(o)}\is N^{\mathbb G}=-
 I_{\Gamma_n}\Bigl({{\Delta M^{\mathbb F}}\over{G_{-}\widetilde G}}\Bigr)\varphi^{(o)}\is N^{\mathbb G},
\end{equation*}
which is a local $\mathbb G$-martingale due to Theorem \ref{theorem2.13}. Thus, again due to
$\cup_n \Gamma_n\cap \Rbrack0,\tau\Rbrack=\Rbrack0,\tau\Rbrack$, we conclude that the martingale $\varphi^{(o)}\is N^{\mathbb G}$ is null,
 or equivalently $\varphi^{(o)}\equiv 0$ $P\otimes D$-a.e. To end the proof of the uniqueness, we will prove that $\widehat M\equiv 0$ implies that $M\equiv 0$ for any $M\in{\cal M}_{0,loc}(\mathbb F)$ satisfying the first and second conditions of (\ref{Condition1}). Indeed, consider such $M$, and remark that in this case we have 
 \begin{eqnarray*}\label{Mhat=0}
 \widehat M=M^{\tau}-\widetilde G^{-1}I_{\Rbrack 0,\tau\Rbrack}\is[m,M]=0.\end{eqnarray*}
 Thus, we get $G_{-}{\widetilde G}^{-1}I_{\Rbrack 0,\tau\Rbrack}\is[M,M]=[M,\widehat M]\equiv 0$ which implies $0=I_{\{\widetilde G>0\}}\is [M,M] \geq  [M,M]^{\widetilde R-}$. Hence, since $\widetilde R\geq R$, $\Delta M I_{\{\widetilde G=0\}}=0$ and $M=M^R$, conclude that $$[M,M]=[M,M]^{\widetilde R}=[M,M]^{\widetilde R-}+\Delta M I_{\Lbrack\widetilde R\Rbrack}=[M,M]^{\widetilde R-}+\Delta M I_{\{\widetilde G=0<G_{-}\}}=0.$$This proves that $M\equiv 0$, and the proof of the theorem is complete.
\end{proof}

 The representation (\ref{MartingaleDecomposGeneral}) was derived in \cite{azemaetal93}, 
 for the Brownian setting and when $\tau$ is an honest time avoiding $\mathbb F$-stopping times. These two features are vital in their proof. Then \cite{blanchetjeanblanc04} extended \cite{azemaetal93} to the case where either all $\mathbb F$-martingales are continuous or $\tau$ avoids $\mathbb F$-stopping times,
  and for a specific family of $\mathbb G$-martingales only. Unfortunately, these assumptions on $(\tau,\mathbb F)$ fail for many popular models in finance and/or insurance. Another attempt for  (\ref{MartingaleDecomposGeneral}) was considered in \cite[Chapitre V,Th\'eor\`eme 5.12]{jeulin80}, where it is proved that the space of square integrable $\mathbb G$-martingales is
 generated by $Y$ given by 
\begin{eqnarray*}\label{Yprocess}
Y=M^{\tau}-G_{-} ^{-1}I_{\Rbrack 0,\tau\Rbrack}\is\langle M,m\rangle^{\mathbb F}+H\is \overline{N}^{\mathbb G}+k\is D+
(L-L^{\tau})+
(1-G_{-})^{-1}I_{\Rbrack\tau,+\infty\Rbrack}\is\langle L,m\rangle^{\mathbb F}.
\end{eqnarray*}
Here $M$ and $L$ are two $\mathbb F$-local martingales, $H\in L^2({\cal P}(\mathbb F), P\otimes D)$ and
$k\in L^2({\rm{Prog}}(\mathbb F), P\otimes D)$ satisfying $\E\left[k_{\tau}\ \big|\ {\cal F}_{\tau-}\right]I_{\{\tau<+\infty\}}=
0,\ P\mbox{-a.s.}$. This result is far less precise than our optional representation. Furthermore, in contrast to our result, the first two local martingale components  in $Y$ are not orthogonal in general. It is worth mentioning that the orthogonality feature  among risks is highly important for risk management, as one can simply deal with each risk individually in that case.\\

 The following illustrates Theorem \ref{TheoRepresentation2} for the popular case when $\mathbb F$ is the complete and right continuous natural filtration of a Brownian motion and a Poisson process.
 \begin{corollary}
Suppose $\mathbb F$ is the augmented filtration of the filtration generated by $(W,p)$, where $W$ is a standard  Brownian motion and $p$ is the Poisson process with intensity one.   Put $N_t^{\mathbb F}:=p_t-t$, and consider a $\mathbb G$-martingale $M^{\mathbb G}$. Then there exists a unique tuple $(\phi,\psi,\varphi^{(o)},\varphi^{(pr)})$ that belongs to the set $L^1_{loc}(W,\mathbb F)\times L^1_{loc}(N^{\mathbb F},\mathbb F)\times {\cal I}^o_{loc}\left(N^{\mathbb G},\mathbb G\right)\times
  L^1_{loc}\left(\widetilde\Omega, {\rm{Prog}}(\mathbb F), P\otimes D\right)$  and satisfies
  $$\left(M^{\mathbb G}\right)^{\tau}=M^{\mathbb G}_0+\phi G_{-}^{-2}I_{\Rbrack0,\tau\Rbrack}\is\widehat{W}+\psi G_{-}^{-2}I_{\Rbrack0,\tau\Rbrack}\is\widehat{N^{\mathbb F}}+\varphi^{(o)}\is N^{\mathbb G}+\varphi^{(pr)}\is D.$$
 \end{corollary}
 
The proof of this corollary follows immediately from combining Theorem \ref{TheoRepresentation2}, and the fact that any $\mathbb F$-local martingale, $M$, there exists a unique pair
 $(\varphi_1,\varphi_2)\in L^1_{loc}(W,\mathbb F)\times L^1_{loc}(N^{\mathbb F},\mathbb F)$ such that $M=M_0+\varphi_1\is W+\varphi_2\is N^{\mathbb F}.$\\

We end this section by showing how Theorem \ref{theo4MartingaleDecomposGeneral} (or Theorem \ref{TheoRepresentation2}) allows us to answer (\ref{QPMM2}).

\begin{corollary}\label{FullsetofPMM} Let $N$ be a $\mathbb G$-local martingale. Then $N$ is a pure default local martingale if and only if 
 there exists a unique pair $(\xi^{(o)},\xi^{(pr)})$ that belongs to ${\cal I}^o_{loc}\left(N^{\mathbb G},\mathbb G\right)\times
  L^1_{loc}\left(\widetilde\Omega, {\rm{Prog}}(\mathbb F), P\otimes D\right) $ satisfying $\E(\xi^{(pr)}_{\tau}|{\cal F}_{\tau})=0$ $P$-a.s. on $\{\tau<+\infty\}$ and 
  \begin{eqnarray*}\label{PMMrepresentation}
  N=N_0+\xi^{(o)}\is N^{\mathbb G}+\xi^{(pr)}\is D.
  \end{eqnarray*}
  As a result, there are only two orthogonal types of  pure default (local) martingales.
\end{corollary}

\begin{proof}
Let $N$ be a pure default local martingale. By stopping, there is no loss of generality in assuming $N$ to be a martingale. Then a direct application of Theorem 2.20 to $N$ leads to the existence of $M\in {\cal M}_{0,loc}(\mathbb F)$, $\varphi^{(o)}\in  {\cal I}^o_{loc}(N^{\mathbb G}, \mathbb G)$ and $\varphi^{(pr)}\in L^1_{loc}({\rm{Prog}}(\mathbb F),P\otimes D) $ such that  
$$N=N_0+G_{-}^{-2}I_{\Rbrack0,\tau\Rbrack}\is{\widehat M}+\varphi^{(o)}\is N^{\mathbb G}+\varphi^{(pr)}\is D,\quad M=M^R,\quad \quad\mbox{and}\quad \Delta M I_{\{\widetilde G=0\}}\equiv 0.$$
Hence the proof of the corollary will be completed as soon as we prove that $\widehat M\equiv 0$. Thus, the rest of the proof concentrates on proving this fact. To this end, for any $\alpha>0$, we consider
$$M^{(\alpha)}:=M-\sum \Delta M I_{\{\vert \Delta M\vert>\alpha\}}+\left(\sum \Delta M I_{\{\vert \Delta M\vert>\alpha\}}\right)^{p,\mathbb F},$$
which is an $\mathbb F$-local martingale with bounded jumps. Since $N$ is a pure default martingale, then $[N, M^{(\alpha)}]$ is a  $\mathbb G$-local martingale, or equivalently $[\widehat M,  M^{(\alpha)}]$ is  a $\mathbb G$-local martingale, or equivalently $[\widehat M,  M-\sum \Delta M I_{\{\vert \Delta M\vert>\alpha\}}]$ is a $\mathbb G$-local martingale. Since $\Delta M I_{\Lbrack\widetilde R\Rbrack}=\Delta M I_{\{\widetilde G=0<G_{-}\}}\equiv0$, we get 
\begin{eqnarray*}
[\widehat M,  M-\sum \Delta M I_{\{\vert \Delta M\vert>\alpha\}}]=G_{-}{\widetilde G}^{-1}I_{\Rbrack0,\tau\Rbrack}I_{\{\vert \Delta M\vert\leq\alpha\}}\is[M,  M],
\end{eqnarray*}
and conclude that $G_{-}{\widetilde G}^{-1}I_{\Rbrack0,\tau\Rbrack}I_{\{\vert \Delta M\vert\leq \alpha\}}\is[M,  M]\equiv 0$ for any $\alpha>0$. By letting $\alpha$ goes to infinity and taking the $\mathbb F$-optional projection afterwards, we deduce that $I_{\{G_{-}>0\}}\is M\equiv 0$, and hence $\widehat M\equiv 0$.  This ends the proof of the corollary.
\end{proof}

%%The proof of this corollary is relegated to the Appendix for the sake of easy exposition.
%%%%%%%%%%%%%%%%%%%%%%%%%%%%%%%%%%%%%%%%%%%%%%%%%%
%%%%%%%%%%
%%%%%%%%%%
%%%%%%%%%%%%%%%%%%%%%%%%%%%%%%%%%%%%%%%%%%%%%%%%%%%%%
\section{Valuation of defaultable securities}\label{subsection4Riskdecomp}
This section constitutes our second main contribution. It answers positively (\ref{Question1}), under some mild condition, by describing the stochastic structures (dynamics) of  some defaultable securities, while allowing the default time $\tau$ to have an arbitrary model. These defaultable securities have their counterpart in life insurance, which are popular mortality securities. Throughout the rest of the paper, a probability measure $Q$ is called {\it  a risk-neutral probability for the model
$(\Omega,\mathbb G)$}, if all discounted price processes of traded securities in the market $(\Omega,\mathbb G)$ are local martingales under $Q$.

\begin{theorem}\label{contractsStructures} Suppose that $P$ is a risk-neutral probability for $(\Omega, \mathbb G)$,  and consider $T\in (0,+\infty)$, $g\in L^1({\cal F}_T)$ and $K\in L^1({\cal O}(\mathbb F), P\otimes D)$. Then the following assertions hold.\\
{\rm{(a)}} The value process of the security with payoff $gI_{\{\tau>T\}}$ at time $T$, is given by
\begin{equation}\label{PureEndowment}
P^{(g)}=P^{(g)}_0+{{I_{\Rbrack 0,\tau\wedge T\Rbrack}}\over{G_{-}}}\is\widehat{M^{(g)}}
-{{M^{(g)}_{-}}\over{G_{-}^2}}I_{\Rbrack 0,\tau\wedge T\Rbrack}\is\widehat m -{{ M^{(g)}}\over{G}}I_{\Rbrack 0,R\Lbrack}\is (N^{\mathbb G})^T,
\  \mbox{with}\ M^{(g)}_t:=\E\left[gG_T\ \big|\ {\cal F}_t\right].
\end{equation}
{\rm{(b)}} The value process of the security with payoff  $K_{\tau}I_{\{\tau\leq T\}}$ at time $T$, is given by
\begin{eqnarray}\label{TermInsurance}
I^{(K)}=I^{(K)}_0+{{I_{\Rbrack 0,\tau\wedge T\Rbrack}}\over{G_{-}}}\is\widehat{M^{(K)}}-
{{Y^{(K)}_{-}}\over{G_{-}^2}}I_{\Rbrack 0,T\wedge\tau\Rbrack}\is\widehat m +
{{KG- Y^{(K)}}\over{G}}I_{\Rbrack 0,T\Rbrack\cap\Lbrack 0,R\Lbrack}\is N^{\mathbb G},
\end{eqnarray}
where
\begin{eqnarray*}\label{M(K)andY(K)}
M^{(K)}_t:=\E\left[\int_0^T K_u dD^{o,\mathbb F}_u\ \big|\ {\cal F}_t\right]\ \ \ \mbox{and}\ \ \ \
Y^{(K)}:=M^{(K)}-K I_{\Lbrack 0,T\Rbrack}\is D^{o,\mathbb F}.\end{eqnarray*}
{\rm{(c)}} The value process of the security with payoff $gI_{\{\tau>T\}}+K_{\tau}I_{\{\tau\leq T\}}$ at time $T$, is given by
\begin{eqnarray*}\label{EndowmentInsurance}
E^{(g,K)}=P^{(g)}+I^{(K)}.
\end{eqnarray*}
Here $P^{(g)}$ and $I^{(K)}$ are given by (\ref{PureEndowment}) and (\ref{TermInsurance}) respectively.\\
 {\rm{(d)}} The value process $B$, of the security with payoff $G_T=P(\tau>T|{\cal F}_T)$ at time $T$, satisfies 
\begin{align}\label{LongevityBond}
B^{\tau}&= B_0+{{I_{\Rbrack 0,\tau\wedge T\Rbrack}}\over{G_{-}}}\is \widehat{M^{(B)}}
-{{M^{(B)}_{-}-{\overline D}^{o,\mathbb F}_{-}}\over{G_{-}^2}}I_{\Rbrack 0,T\wedge\tau\Rbrack}\is\widehat m
 +{{\xi^{(G)}G- M^{(B)}+{\overline D}^{o,\mathbb F}}\over{G}}I_{\Rbrack 0,R\Lbrack}I_{\Rbrack 0,T\Rbrack}\is N^{\mathbb G}\nonumber\\
&\quad +\left(\E[G_T\ \big|\ {\cal G}_{\tau}]-\xi^{(G)}_{\tau}\right) I_{\Lbrack\tau, +\infty\Lbrack},
\end{align}
where
\begin{equation}\label{xiGandM(G)}
M^{(B)}_t:=\E\left[{\overline D}^{o,\mathbb F}_{\infty}-{\overline D}^{o,\mathbb F}_0\ \big|\ {\cal F}_{t\wedge T}\right],
\ \ \ \ \xi^{(G)}:={{d{\overline D}^{o,\mathbb F}}\over{d{D}^{o,\mathbb F}}},\ \ \ \ {\overline D}^{o,\mathbb F}:=
\left(G_T I_{\Lbrack\tau, +\infty\Lbrack}\right)^{o,\mathbb F}.
\end{equation}
\end{theorem}

\begin{proof} This proof contains two parts. The first part proves assertions (a), (b) and (c), while the last part deals with assertion (d).\\
{\bf Part 1:} Remark that
$$
gI_{\{\tau>T\}}=h_{\tau},\ \ \ \ \mbox{where}\ \ h_t:=gI_{\Rbrack T,+\infty\Lbrack}(t)\ \mbox{is an $\mathbb F$-predictable process}.
$$
Thus, we get $P^{(g)}_t=\E[h_{\tau}\ \big|\ {\cal G}_t]$. As a result, we deduce that
$$P^{(g)}_t=P^{(g)}_{t\wedge\tau}=P^{(g)}_{t\wedge\tau\wedge T},\ \ \ \ \ \ \ h^T\equiv 0,\quad I_{\Rbrack 0,T\Rbrack} (h\is D^{o,\mathbb F})_{-}\equiv 0 \quad \mbox{and}\quad (h\is D^{o,\mathbb F})^T\equiv 0.$$
Therefore, by inserting these in (\ref{RepresentationofH}) and using $M^h_{t\wedge T}=M^{(g)}_t:=\E[gG_T\ \big|\ {\cal F}_t]$, assertion (a) follows immediately. 
Similarly, assertion (b) follows from combining Theorem {\ref{TheoRepresentation}}-(a), for the payoff process $h$ taking the form of
$h_t:=K_tI_{\Lbrack0,T\Rbrack}(t)$,  and the fact that $I^{(K)}_t=I^{(K)}_{t\wedge T}$.     It is clear that assertion (c) is a direct consequence of assertions (a) and (b).\\
{\bf Part 2:}  Herein, we prove assertion (d).  To this end, we put 
$$k_{\tau}:=\E[G_T\big|{\cal G}_{\tau}]-\E[G_T\big|{\cal F}_{\tau}]\quad\mbox{and}\quad {\overline D}:=G_TI_{\Lbrack\tau,+\infty\Lbrack},$$
and focus on proving 
\begin{equation}\label{equa3001}
\xi^{(G)}_{\tau}=\E[G_T\big|{\cal F}_{\tau}]\ \ P\mbox{-a.s. on}\quad\{\tau<+\infty\}.\end{equation}
Thus, we consider $O\in{\cal O}(\mathbb F)$ and  derive
\begin{align*}
\E\left[G_T I_O(\tau)I_{\{\tau<+\infty\}}\right]&=\E\left[\int_0^{+\infty} I_O(t) d{\overline D}_t\right]=
\E\left[\int_0^{+\infty} I_O(t) d{\overline D}^{o,\mathbb F}_t\right]\\
&=\E\left[\int_0^{+\infty} I_O(t)\xi^{(G)}_t d{D}^{o,\mathbb F}_t\right]=\E\left[I_O(\tau)\xi^{(G)}_{\tau}I_{\{\tau<+\infty\}}\right].
\end{align*}
This proves (\ref{equa3001}). Thus, by applying Theorem \ref{TheoRepresentation2} to the $\mathbb G$-martingale $M^{\mathbb G}=B:=\ ^{o,\mathbb G}(G_T)$ for which  $h=\xi^{(G)}$ (since $k_{\tau}+\xi^{(G)}_{\tau}=E(G_T|{\cal G}_{\tau})=M^{\mathbb G}_{\tau}$ on $\{\tau<+\infty\}$) and  using the fact that $B_{t\wedge\tau}= B_{t\wedge\tau\wedge T}$, the proof of assertion (d) follows immediately.  This ends the proof of the theorem.
\end{proof}

As it was aforementioned, the four defaultable securities of Theorem \ref{contractsStructures}  have their mortality and/or longevity securities counterparts. In fact, in life insurance, the security with payoff $gI_{\{\tau>T\}}$ is known as a pure endowment insurance with benefit $g$. This insurance contract pays $g$ dollars at time $T$ if the insured survives.  In life insurance, the security with payoff  $K_{\tau}I_{\{\tau\leq T\}}$ is called a term insurance contract with benefit process. This contract pays $K_{\tau}$ at $\tau$ if the insured dies before or at the term of the contract. The security with payoff $G_T=P(\tau>T|{\cal F}_T)$ is very popular in life insurance and is called a zero-coupon longevity bond. It is an insurance contract that pays the conditional survival probability at term $T$, and it plays important role in the securitization process. The security with  payoff $gI_{\{\tau>T\}}+K_{\tau}I_{\{\tau\leq T\}}$ is known as endowment insurance contract with benefit pair $(g,K)$.

\begin{remark}\label{Remark3.3}
{\rm{(a)}} The stochastic structures of the securities' value processes described in Theorem \ref{contractsStructures}
allow us to single out all types of risks that every security possesses. This is very important for the risk management of the extra uncertainty borne by $\tau$. In life insurance, this is important for the mortality and/or longevity securitization process. In fact, defaultable securities with no second type 
of pure default risk will be irrelevant in reducing this type of risk in the risk management process.\\
 {\rm{(b)}}  By comparing (\ref{PureEndowment}), (\ref{TermInsurance}) and (\ref{LongevityBond}), we conclude that both the first and the second type of defaultable securities possess the same type of risks,
 while the third type of defaultable security bears, in addition, the pure mortality risk of type two given by
 $\left(\E[G_T\ \big|\ {\cal G}_{\tau}]-\xi^{(G)}_{\tau}\right)I_{\Lbrack\tau, +\infty\Lbrack}$. As proved in the previous subsection, this
 risk is orthogonal to the other risks (i.e. the financial risk and the pure mortality risk of the first type). This second type of pure mortality risk, in the longevity bond, vanishes if and only if
 $$\E[G_T\ \big|\ {\cal G}_{\tau}]I_{\{\tau<T\}}=\E[G_T\ \big|\ {\cal F}_{\tau}]I_{\{\tau<T\}},\ \ \ \ \ P\mbox{-a.s.},$$
  due to the fact that we always have
  $$\E[G_T\ \big|\ {\cal G}_{\tau}]I_{\{\tau\geq T\}}=\E[G_T\ \big|\ {\cal F}_{\tau}]I_{\{\tau\geq T\}},\ \ \ \ \ P\mbox{-a.s.}.$$
  This means that this risk occurs only when  the random time (default or death) occurs before the maturity $T$.  In general, this risk vanishes when $\tau$ avoids $\mathbb F$-stopping times or when
   ${\cal G}_{\tau}(={\cal F}_{\tau+})$ coincides with ${\cal F}_{\tau}$. Thus, imposing these assumptions, as in \cite{blanchetjeanblanc04}, 
   boils down to neglecting this type of risk.\\
{\rm{(c)}} The assumption on the probability $P$ in Theorem \ref{contractsStructures} and in the following corollary is not a restriction in our view.
   It is assumed for the sake of easy exposition only, as one can calculate every process used in
   the theorem (starting with the processes $G,\widetilde G, G_{-}$) under a chosen risk-neutral measure, $Q$, for the informational model
   $(\Omega,\mathbb G)$. \end{remark}

\noindent The rest of this section illustrates Theorem \ref{contractsStructures} for the case of a pseudo-stopping time $\tau$.

\begin{corollary}\label{pseudocase} Suppose that $P$ is a risk neutral probability for $(\Omega, \mathbb G)$, and $\tau$ is a
pseudo-stopping time satisfying $G>0$. Then the processes $P^{(g)}$, $I^{(K)}$ and $B^{\tau}$ of Theorem \ref{contractsStructures} take the following forms.
\begin{align}
P^{(g)}&= P^{(g)}_0+{{I_{\Rbrack0,\tau\Rbrack}}\over{G_{-}}}\is M^{(g)}-{{ M^{(g)}}\over{G}}\is \left( N^{\mathbb G}\right)^T,\label{Pg4pseudo}\\
I^{(K)}&= I^{(K)}_0+{{I_{\Rbrack0,\tau\Rbrack}}\over{G_{-}}}\is M^{(K)}+
{{KG- Y^{(K)}}\over{G}}\is \left( N^{\mathbb G}\right)^T,\nonumber\\
%\label{Ik4pseudo}\\
B^{\tau}&= B_0+{{I_{\Rbrack0,\tau\Rbrack}}\over{G_{-}}}\is M^{(B)}
+\left[\xi^{(G)}+{{- M^{(B)}+{\overline D}^{o,\mathbb F}}\over{G}}\right]\is\left( N^{\mathbb G}\right)^T+
\left[\E[G_T\ \big|\ {\cal G}_{\tau}]-\xi^{(G)}_{\tau}\right]I_{\Lbrack\tau, +\infty\Lbrack},\nonumber
% \label{B4pseudo}
\end{align}
where $\left(M^{(B)},\xi^{(G)}, \overline{D}^{o,\mathbb F}\right)$ is given by (\ref{xiGandM(G)}).
\end{corollary}

\noindent The proof of the corollary follows immediately from combining Theorem \ref{contractsStructures}
with the fact that $m\equiv m_0$ whenever $\tau$ is a pseudo-stopping time, and will be omitted herein.\\

%\begin{corollary}\label{independCase}
 \noindent  If $\tau$ is independent of ${\cal F}_{\infty}$ such that $P(\tau>T)>0$, then on the one hand (\ref{Pg4pseudo}) becomes
 \begin{eqnarray*}\label{Pg}
 P^{(g)}=P^{(g)}_0-{{ gP(\tau>T)}\over{P(\tau>\cdot)}}\is \left( N^{\mathbb G}\right)^T=
 P^{(g)}_0-{{ gP(\tau>T)}\over{P(\tau>\cdot)}}\is\left({\overline N}^{\mathbb G}\right)^T.
 \end{eqnarray*}
 On the other hand, the security with payoff $G_T$ has a constant value process equal to $G_T$, and hence it can not be used for hedging any risk! Thus,  under the independence condition between $\tau$ and $\mathbb F$, the security with payoff $I_{\{\tau>T\}}$, i.e. the contract pays one dollar to the beneficiary if no default occur over the term $[0,T]$) is more adequate to hedge defaultable liabilities.  

\begin{appendices}
	
%\section{A Radon-Nikodym property}	
%%%%%%%%%%%%%%%%%%%%%%%%%%%%%%%%%%%%%%%%%%%%%%%%%%%%%%%%%%%%%%%
%%%%%%%%%%%%%%%%%%%%%%
%%%%%%%%%%%%%%%%%%%%%%
%%%%%%%%%%%%%%%%%%%%%%%%%%%%%%%%%%%%%%%%%%%%%%%%%%%%%%%%%%%%%%%%%%%%

%%%\section{Proof of Proposition \ref{NbarG-NG}}\label{AppendixproofProposition2.3}
%%\qed
%%%XXXXXXXXXXXXXXXXXXXXXXXXXXXXXX
%%%%XXX
%%%%XXXXXXXXXXXXXXXXXXX
%%%%XXXXXXXXXXXXXXXXXXXX

\section{Proofs of Lemma \ref{LemmaTech} and  Theorem \ref{TheoRepresentation}}\label{SubsectionproofTheorem1}
This subsection is devoted to the proof of this main theorem, and its technical lemma. Thus, throughout this subsection, we consider
 $h\in L^1\left({\cal O}(\mathbb F),P\otimes D\right)$ and the associated $\mathbb G$-martingale
  $H_t:=\E[h_{\tau}\big|{\cal G}_t]$. Then remark that we can decompose this martingale as follows
\begin{equation}\label{DecompHexplicit}
H_t = h_\tau I_{\{\tau\leq t\}} +
{{ I_{\{\tau>t\}} }\over{G_t}} \E\left[h_\tau  I_{\{\tau>t\}} |{\cal F}_t\right] = (h\is D)_t +J^h_tI_{\Lbrack 0,\tau\Lbrack}(t)=\left((h-J^h)\is D\right)_t+J^h_{\tau\wedge t},
\end{equation}
where the process $J^h$ is defined by
\begin{equation}\label{Jprocess}
J^h := {{Y}\over{K}},\quad\quad Y_t := \E\left[h_\tau  I_{\Lbrack 0,\tau\Lbrack} (t) \big|{\cal F}_t\right],\quad\quad
K:=G+(G_{R-}+I_{\{G_{R-}=0\}})I_{\Lbrack R,+\infty\Lbrack}.
\end{equation}
It is easy to notice that $Y$ is  is a $\mathbb F$-semimartingale and satisfies
\begin{equation}\label{processY}
Y=M^h-h\is D^{o,\mathbb F},\end{equation}
where $M^h$ is defined in (\ref{processesMhandJ}), that we recall herein by $M^h:=\ ^{o,\mathbb F}(\int_0^{+\infty}h_udD^{o,\mathbb F}_u)$. 

\begin{proof}[Proof of Lemma  \ref{LemmaTech}:]  Due to $L^1_{loc}({\cal O}(\mathbb F),P\otimes D)\subset {\cal I}^o_{loc}(N^{\mathbb G}, \mathbb G)$ and 
$G^{-1}(M^h-h\is D^{o,\mathbb F})I_{\Lbrack0,R\Lbrack}=J^h I_{\Lbrack0,R\Lbrack}$, we remark that the proof of the lemma boils down to proving 
 $J^hI_{\Lbrack 0,R\Lbrack}\in {\cal I}_{\loc}^o(N^{\mathbb G},\mathbb G)$ . To prove this latter fact, we consider the sequence of
  $\mathbb F$-stopping times $(\sigma_n)_{n\geq 1}$ given by
$$
\sigma_n:=\inf\{t\geq 0\ :\ \vert J^h_t\vert>n \},\quad n\geq 1.$$
Since $J^h$ is a RCLL and $\mathbb F$-adapted process with real values, then the sequence $(\sigma_n)_{n\ge 1}$
 increases to infinity almost surely.  Then we calculate
\begin{align*}
	\E\Bigg[{{\vert J^h\vert G}\over{\widetilde G}}I_{\{\widetilde G>0\}}\is D_{\sigma_n}\Bigg]&\leq n
+\E\Bigg[{{\vert J^h_{\sigma_n}\vert G_{\sigma_n}}\over{\widetilde G_{\sigma_n}}}I_{\{\tau=\sigma_n<+\infty\}}\Bigg]
	=n +\E\Bigg[{{\vert J^h_{\sigma_n}\vert G_{\sigma_n}}\over{\widetilde G_{\sigma_n}}}(\widetilde G_{\sigma_n}-
G_{\sigma_n})I_{\{\widetilde G_{\sigma_n}>0\}}\Bigg]\\
	\\
	&\leq n +\E\big[\vert Y_{\sigma_n}\vert\big]\leq n +\E\big[\vert h_{\tau}\vert I_{\{\sigma_n<\tau<+\infty\}}\big]<+\infty . \end{align*}
This proves that $J^hI_{\Lbrack 0,R\Lbrack}\in {\cal I}_{\loc}^o(N^{\mathbb G},\mathbb G)$, and ends the proof of the lemma.
\end{proof}
 The rest of this section focuses on proving Theorem  \ref{TheoRepresentation}. This proof relies heavily on understanding the dynamics of the process $K$ and subsequently that of $J^h$.
 The following lemma addresses useful properties, of the process $K$, that will be used throughout the proof of the theorem.

\begin{lemma}\label{lemma4proof1} Let $K$ be given in (\ref{Jprocess}). Then the following assertions hold.\\
{\rm{(a)}} $K^{\tau}$ is a positive $\mathbb G$-semimartingale satisfying the following
\begin{equation}\label{prop4Ktau}
K^{\tau}=G^{\tau}+G_{-}I_{\Lbrack R\Rbrack}\is D,\ \ \  K_{-}^{\tau}=
G_{-}^{\tau},\ \ \inf_{t\geq 0}K_{t\wedge\tau}>0,\ \ KI_{\Lbrack 0,R\Lbrack}+KI_{\Lbrack \tau\Rbrack\cap\Lbrack R\Rbrack}=
GI_{\Lbrack 0,R\Lbrack}+G_{-}I_{\Lbrack \tau\Rbrack\cap\Lbrack R\Rbrack}.\end{equation}
{\rm{(b)}}  As a result, $\left(K^{\tau}\right)^{-1}$ is a positive $\mathbb G$-semimartingale admitting the following decomposition.
\begin{align}
d\left({1\over{K^{\tau}}}\right)&= -\left(G_{-}^{\tau}\right)^{-2}dm^{\tau}+(GG_{-}^2)^{-1} I_{\Rbrack0, \tau\Lbrack}d[m,m]+
  (G_{-}-\Delta m)(GG_{-}^2)^{-1}I_{\Lbrack0, \tau\Lbrack}dD^{o,\mathbb F}\nonumber\\
&\quad +\left\{{{G\Delta m-G_{-}\Delta G}\over{GG_{-}^2}}I_{\Rbrack0, R\Lbrack}
+{{\Delta m}\over{G_{-}^2}}I_{\Lbrack R\Rbrack}\right\}dD.\label{Decomp4K-1tau}
\end{align}
{\rm{(c)}}  For any $\mathbb G$-semimartingale $L$, we have
\begin{equation}\label{quadraticform}
d\Big[L,{1\over{K^{\tau}}}\Big]=-{1\over{GG_{-}}}I_{\Rbrack0, \tau\Lbrack}d[L,m]
+{{\Delta L}\over{GG_{-}}}I_{\Rbrack0, \tau\Lbrack}dD^{o,\mathbb F}-{{\Delta L\Delta G}\over{GG_{-}}}I_{\Rbrack0, R\Lbrack}dD.
\end{equation}
{\rm{(d)}} On $\{\widetilde R<+\infty\}$, we have
\begin{eqnarray*}\label{2termscanceled}
\Delta M^h_{\widetilde R}-J^h_{\widetilde R-}\Delta m_{\widetilde R}=0,\ \ \ P\mbox{-a.s.}.\end{eqnarray*}
\end{lemma}
%%%%%%%%%%%%%%%%%%%%%%%%XXXXXXXXXXXXXXXXXXXXXXXXXXXXXXXX
\begin{proof} The proof is achieved in four steps, where we prove assertions (a), (b), (c) and (d) respectively.\\
{\bf Step 1:} Thanks to \cite[XX.14]{dellacheriemeyer92} , we have
$$\Lbrack0,\tau\Rbrack\subset\{G_{-}>0\}\cap\{\widetilde G>0\}\ \ \ \ \mbox{and}\ \ \ \ \ \ \ \ \tau\leq R\ \ \ P\text{-a.s}.$$
As a result, we get
\begin{align*}
K^{\tau}&=G^{\tau}+\left[G_{R-}+I_{\{G_{R-}=0\}}\right]I_{\{\tau=R\}}I_{\Lbrack R,+\infty\Lbrack}
=G^{\tau}+G_{R-}I_{\{\tau=R\}}  I_{\Lbrack R,+\infty\Lbrack}\\
\\
& =G^{\tau}+G_{-}I_{\Lbrack R\Rbrack}\is I_{\Lbrack \tau,+\infty\Lbrack} =G^{\tau}+G_{-}I_{\Lbrack R\Rbrack}\is D.
\end{align*}
This proves the first equality in (\ref{prop4Ktau}). The proofs of the second and the last equalities in (\ref{prop4Ktau})
 follow immediately from the facts that $K=G$ on $\Rbrack0,R\Lbrack\supset\Rbrack0,\tau\Lbrack$ and $K=G_{-}$ on $\Lbrack R\Rbrack\cap\Lbrack\tau\Rbrack$. Again, thanks to these two facts and the first equality in (\ref{prop4Ktau}), we deduce that $K^{\tau}$ is a $\mathbb G$-semimartingale such that $K^{\tau}>0$ and $K^{\tau}_{-}>0$. Hence, $\displaystyle\inf_{t\geq 0}K^{\tau}_t>0\ \ P$-a.s. and the proof of assertion (a) is completed.\\
 %%%%%%%%%%%%%%%%%%%%%%%%%%%%%%%%%%%%%%
{\bf Step 2:}  It is obvious, from assertion (a), that $(K^{\tau})^{-1}$ is a well-defined positive $\mathbb G$-semimartingale.
 Then a direct application of Ito's formula leads to
\begin{equation}\label{ItoforK-1}
d\left({1\over{K^{\tau}}}\right)=-{1\over{(K^{\tau}_{-})^2}}dK^{\tau}+{1\over{K^{\tau}(K^{\tau}_{-})^2}}d[K^{\tau}].
\end{equation}
Thanks to (\ref{prop4Ktau}), $(\Delta G)^2I_{\Lbrack R\Rbrack}=G_{-}^2 I_{\Lbrack R\Rbrack}$ and $[G,G]=[m,m]-(\Delta G+\Delta m)\is D^{o,\mathbb F}$,  we derive
$$
d[K^{\tau},K^{\tau}]=I_{\Rbrack0,\tau\Lbrack}d [m,m]-(\Delta G+\Delta m)I_{\Rbrack0,\tau\Lbrack}dD^{o,\mathbb F}+
(\Delta G)^2I_{\Rbrack0,R\Lbrack}dD.$$
Thus, by inserting this equality together with $KI_{\Rbrack0,R\Lbrack}=GI_{\Rbrack0,R\Lbrack}$ and $K_{-}I_{\Rbrack0,R\Lbrack}=G_{-}I_{\Rbrack0,R\Lbrack}$, in (\ref{ItoforK-1}), the proof of assertion (b)
follows immediately.\\
%%%%%%%%%%%%%%%%%%%%%%%%%%%%%%%%%%%%%%%%%%%%%%
{\bf Step 3:}  Let $L$ be a $\mathbb G$-semimartingale. Then, by using (\ref{ItoforK-1}),
\[
\Big[L,\frac{1}{K^\tau}\Big] = -\frac{1}{K^\tau K_-^\tau} \is  [L,K^\tau],
\]
and
$$
[L,K^{\tau}]=I_{\Lbrack0,\tau\Lbrack}\is [m,L]  - I_{\Lbrack 0, \tau \Lbrack } \is [L,D^{o,\mathbb F}]  +
(\Delta G)(\Delta L)I_{\Rbrack0,R\Lbrack}\is D,$$
 we easily derive (\ref{quadraticform}). \\
 %%%%%%%%%%%%%%%%%%%%%%%%%%%%%%%%%%%%%%%
 {\bf Step 4:}  Herein, we prove assertion  (d).  To this end, we remark that $\{\widetilde R<+\infty\}=\{\widetilde R=R<+\infty\}$, $Y=GJ^h$ on $\Rbrack0,R\Lbrack$ which implies $Y_{-}=G_{-}J_{-}^h$ on $\Rbrack0,R\Rbrack$, and $Y_{\widetilde R}=\Delta D^{o,\mathbb F}_{\widetilde R}=0$ $P$-a.s. on $\{\widetilde R<+\infty\}$. By combining all these with (\ref{processY}) and $\Delta m=\widetilde G-G_{-}$, on $\{\widetilde R<+\infty\}$ we get
 $$
 \Delta M^h_{\widetilde R}-J^h_{\widetilde R-}\Delta m_{\widetilde R}=\Delta Y_{\widetilde R}+h_{\widetilde R}\Delta D^{o,\mathbb F}_{\widetilde R}+J^h_{\widetilde R-}G_{\widetilde R-}=\Delta Y_{\widetilde R}+J^h_{\widetilde R-}G_{\widetilde R-}=- Y_{\widetilde R-}+J^h_{\widetilde R-}G_{\widetilde R-}=0.$$
 This ends the proof of the lemma.
\end{proof}

Now, we are in the stage of proving Theorem \ref{TheoRepresentation}. 

\begin{proof}[Proof of Theorem \ref{TheoRepresentation}] The proof of the theorem will be given in three steps where
 we prove the three assertions respectively.\\
 {\bf Step 1:} Here, we prove assertion (a).  Thanks to Lemma \ref{LemmaTech}, $J^hI_{\Lbrack 0,R\Lbrack}\in {\cal I}_{\loc}^o(N^{\mathbb G},\mathbb G)$, 
 and  remark that $Y$, defined in (\ref{Jprocess}), is a $\mathbb G$-semimartingale and satisfies (\ref{processY}). Then, thanks to Ito's calculus, we derive
\begin{equation}\label{Jtau1}\displaystyle
d(J^h)^{\tau}=d\left({{Y^{\tau}}\over{K^{\tau}}}\right)={1\over{K_{-}^{\tau}}}dY^{\tau}+Y_{-}^{\tau}d\left({1\over{K^{\tau}}}\right)+
d\displaystyle\left[{1\over{K^{\tau}}},Y^{\tau}\right].
\end{equation}
Thus, the proof of assertion (a) of the theorem boils down to calculating each of the three terms in the RHS of the above equality,
 and to simplifying them afterwards.\\
  By combining $dY^{\tau}=d(M^h)^{\tau}-hI_{\Rbrack 0,\tau\Rbrack}dD^{o,\mathbb F}$,
  (\ref{prop4Ktau}) and (\ref{processMhat}), we write

\begin{align}\label{Term1}
{{1}\over{K_{-}^{\tau}}}dY^{\tau} &= {{1}\over{G_{-}^{\tau}}}d{\widehat M^h}
+{{1}\over{{\widetilde G}G_{-}^{\tau}}}I_{\Rbrack 0,\tau\Rbrack}d[m,M^h]
 -{{hI_{\Rbrack 0,\tau\Rbrack}}\over{G_{-}^{\tau}}}dD^{o,\mathbb F}
  -{1\over{G_{-}^\tau}}I_{\Rbrack 0,\tau\Rbrack}d\left(\Delta M^h_{\widetilde R}I_{\Lbrack\widetilde R,
  +\infty\Lbrack}\right)^{p,\mathbb F}\nonumber \\
&= {{1}\over{G_{-}^{\tau}}}d{\widehat M^h}+{{1}\over{{\widetilde G}G_-}}I_{\Rbrack 0,\tau\Lbrack}d[m,M^h]
-{{hI_{\Rbrack 0,\tau\Lbrack}}\over{G_-}}dD^{o,\mathbb F}
 -{1\over{G_{-}}}I_{\Rbrack 0,\tau\Rbrack}d\left(\Delta M^h_{\widetilde R}I_{\Lbrack\widetilde R,+\infty\Lbrack}\right)^{p,\mathbb F} \\
& \quad +\Bigl[{{\Delta m \Delta M^h}\over{{\widetilde G}G_{-}^{\tau}}}-{{ h \Delta D^{o,\mathbb F}}\over{G_{-}^{\tau}}}\Bigr]d D  . \nonumber
\end{align}

Thanks to (\ref{Decomp4K-1tau}) and again (\ref{processMhat}) (recall that $Y^{\tau}_{-}/K^{\tau}_{-}=Y^{\tau}_{-}/G^{\tau}_{-}=(J^h)^{\tau}_{-}$),
 we calculate

\begin{align}\label{Term2}
& Y_{-}^{\tau} d\left({1\over{K^{\tau}}}\right)= -{{(J^h)^{\tau}_{-}}\over{G_{-}^{\tau}}}dm^{\tau}+
{{(J^h)^{\tau}_{-}}\over{GG_{-}^{\tau}}}I_{\Rbrack 0,\tau\Lbrack}d[m,m]+
{{(J^h)^{\tau}_{-}(G_{-}-\Delta m)}\over{GG_{-}^{\tau}}}I_{\Rbrack 0,\tau\Lbrack}dD^{o,\mathbb F}\nonumber\\
\nonumber\\
&\qquad \qquad \qquad \ +\left[{{(J^h)^{\tau}_{-}(G\Delta m-G_{-}\Delta G)}\over{GG_{-}^{\tau}}}I_{\Rbrack 0,R\Lbrack}+
{{(J^h)^{\tau}_{-} \Delta m}\over{G_{-}^{\tau}}}I_{\Lbrack R\Rbrack}\right]dD\nonumber\\
\nonumber\\
& =-{{(J^h)^{\tau}_{-}}\over{G_{-}^{\tau}}}d{\widehat m}+
{{J^h_{-}}\over{G_{-}}}I_{\Rbrack 0,\tau\Rbrack}d\left(\Delta m_{\widetilde R}I_{\Lbrack\widetilde R,+\infty\Lbrack}\right)^{p,\mathbb F}
+{{J^h_{-}(\Delta m)^2}\over{\widetilde G GG_{-}}}I_{\Rbrack 0,\tau\Lbrack}dD^{o,\mathbb F}
+{{J^h_{-}(G_{-}-\Delta m)}\over{GG_{-}}}I_{\Rbrack 0,\tau\Lbrack}dD^{o,\mathbb F}\nonumber\\
\nonumber\\
&\quad   +\left[{{(J^h)^{\tau}_{-}(G\Delta m-G_{-}\Delta G)}\over{GG_{-}^{\tau}}}I_{\Rbrack 0,R\Lbrack}-
{{(J^h)_{-}^{\tau}(\Delta m)^2}\over{\widetilde G G_{-}^{\tau}}}+{{(J^h)_-^\tau\Delta m}\over{G_{-}^{\tau}}}I_{\Lbrack R\Rbrack}\right]dD\nonumber\\
\nonumber\\
& =-{{(J^h)^{\tau}_{-}}\over{G_{-}^{\tau}}}d{\widehat m}+
{{J^h_{-}}\over{G_{-}}}I_{\Rbrack 0,\tau\Rbrack}d\left(\Delta m_{\widetilde R}I_{\Lbrack\widetilde R,+\infty\Lbrack}\right)^{p,\mathbb F}
+{{J^h_{-}G_{-}}\over{\widetilde G G}}I_{\Rbrack 0,\tau\Lbrack}dD^{o,\mathbb F}\nonumber\\
\nonumber\\
&\quad +\left[{{(J^h)^{\tau}_{-}(G\Delta m-G_{-}\Delta G)}\over{GG_{-}^{\tau}}}I_{\Rbrack 0,R\Lbrack}-
{{(J^h)_{-}^{\tau}(\Delta m)^2}\over{\widetilde G G_{-}^{\tau}}}+{{(J^h)_-^\tau\Delta m}\over{G_{-}^{\tau}}}I_{\Lbrack R\Rbrack}\right]dD.
\end{align}

By applying (\ref{quadraticform}) to $L=Y^{\tau}=(M^h)^{\tau}-(h\is D^{o,\mathbb F})^{\tau}$, the last term in the RHS of (\ref{Jtau1}) becomes
\begin{align}\label{Term3}
d\left[Y^{\tau} ,{1\over{K^{\tau}}}\right]&=-\frac{1}{G G_-} I_{\Rbrack 0, \tau \Lbrack } d [Y,m] +
 \frac{\Delta Y}{G G_-} I_{\Rbrack 0, \tau \Lbrack } d D^{o,\mathbb F} -
 \frac{\Delta Y \Delta G}{G G_-} I_{\Rbrack 0, R \Lbrack} d D \nonumber\\
\nonumber\\
&=-{{1}\over{GG_{-}}}I_{\Rbrack 0,\tau\Lbrack}d[M^h,m]+{{\Delta Y+h\Delta m}\over{GG_{-}}}I_{\Rbrack 0,\tau\Lbrack} dD^{o,\mathbb F}-
\frac{\Delta Y \Delta G}{G G_-} I_{\Rbrack 0, R \Lbrack} d D .
\end{align}
Thanks to Lemma \ref{lemma4proof1}-(d), we conclude that
\[
-{1\over{G_{-}}}I_{\Rbrack 0,\tau\Rbrack} \is \left(\Delta M^h_{\widetilde R}I_{\Lbrack\widetilde R,+\infty\Lbrack}\right)^{p,\mathbb F}
 + {{J^h_{-}}\over{G_{-}}}I_{\Rbrack 0,\tau\Rbrack}\is\left(\Delta m_{\widetilde R}I_{\Lbrack\widetilde R,
 +\infty\Lbrack}\right)^{p,\mathbb F} = 0.
\]
By taking this equality into consideration, after inserting (\ref{Term1}), (\ref{Term2}) and (\ref{Term3}) in (\ref{Jtau1}), we get
\begin{align}\label{JtauMainEquation}
d(J^h)^{\tau} & = {1\over{G_{-}^{\tau}}}d{\widehat M^h} -{{(J^h)^{\tau}_{-}}\over{G_{-}^{\tau}}}d{\widehat m} +
 \left\{{{\Delta Y+h\Delta m}\over{GG_{-}}} + {{J^h_{-}G_{-}}\over{\widetilde G G}}
-{{h}\over{G_{-}}}-{{\Delta m\Delta M^h}\over{\widetilde G G G_{-}}}\right\}I_{\Rbrack 0,\tau\Lbrack}dD^{o,\mathbb F}\nonumber\\
\nonumber\\
&\quad +\left[{{(J^h)^{\tau}_{-}(G\Delta m-G_{-}\Delta G)}\over{GG_{-}^\tau}}I_{\Rbrack 0,R\Lbrack}-{{(J^h)_{-}^{\tau}(\Delta m)^2-
\Delta m\Delta M^h}\over{\widetilde G G_{-}^{\tau}}}+{{(J^h)_{-}^{\tau}\Delta m}\over{G_{-}^{\tau}}}I_{\Lbrack R\Rbrack} \right. \nonumber \\
& \qquad \quad  \left. -{{h\Delta D^{o,\mathbb F}}\over{G_{-}^{\tau}}}-\frac{\Delta Y \Delta G}{G G_-} I_{\Rbrack 0, R \Lbrack}\right]dD\nonumber\\
\nonumber\\
 &=: {1\over{G_{-}^{\tau}}}d\widehat{M^h} -{{(J^h)^{\tau}_{-}}\over{G_{-}^{\tau}}}d{\widehat m}+
 \xi^{(1)}I_{\Rbrack 0,\tau\Lbrack}dD^{o,\mathbb F}+\left[\xi^{(2)}I_{\Rbrack 0,R\Lbrack}+\xi^{(3)}I_{\Lbrack R\Rbrack}\right]dD.
 \end{align}

 Now, we need to simplify the expressions $\xi^{(i)}$ for $i=1,2,3$. In fact, on $\Rbrack 0,\tau\Lbrack$, we calculate
 \begin{equation}\label{xi1}
 \xi^{(1)}={{\Delta Y+h\Delta m}\over{GG_{-}}} + {{J^h_{-}G_{-}}\over{\widetilde G G}}
 -{{h}\over{G_{-}}}-{{\Delta m\Delta M^h}\over{\widetilde G G G_{-}}}={{J^h-h}\over{\widetilde G}}.
 \end{equation}

 Similarly, on $\Rbrack 0,R\Lbrack\cap\Rbrack 0,\tau\Rbrack$, we use $\Delta Y=GJ^h-G_{-}J^h_{-}$, $\Delta M^h=
 \Delta Y+h\Delta D^{o,\mathbb F}$, $\Delta D^{o,\mathbb F}=\widetilde G-G$, and $\Delta m=\widetilde G-G_{-}$, and we derive
 \begin{align}\label{xi2}
 \xi^{(2)}&=\Bigl[{{J^h_{-}(G\Delta m-G_{-}\Delta G)}\over{GG_{-}}}-
 {{h\Delta D^{o,\mathbb F}}\over{G_{-}}}-{{\Delta Y\Delta G}\over{GG_{-}}}\Bigr]-{{J^h_{-}(\Delta m)^2-\Delta m\Delta M^h}\over{\widetilde G G_{-}}}\nonumber\\
 &={{{\widetilde G}G(J^h_{-}-h)-G^2(J^h-h)+  GG_{-}\Delta J^h}\over{GG_{-}}}-{{\Delta m\left[\widetilde G(J^h_{-}-h)-
 G(J^h-h)\right]}\over{\widetilde G G_{-}}}
 \nonumber\\
 &={{J^h-h}\over{\widetilde G}}\Delta D^{o,\mathbb F}.
 \end{align}
By using $YI_{\Lbrack R\Rbrack}=0$ (since $\tau\leq R\ P$-a.s.) and $\Delta D^{o,\mathbb F}I_{\Lbrack R\Rbrack}=\widetilde GI_{\Lbrack R\Rbrack}$, on
 $\Lbrack R\Rbrack\cap\Rbrack 0,\tau\Rbrack$, we get

 \begin{align}\label{xi3}
 \xi^{(3)}&=-{{J^h_{-}(\Delta m)^2-\Delta m\Delta M^h}\over{\widetilde G G_{-}}}-{{h\Delta D^{o,\mathbb F}}\over{G_{-}}}+
 {{J^h_{-}\Delta m}\over{G_{-}}}\nonumber\\
 &={{-\Delta m\widetilde G(J^h_{-}-h)+{\widetilde G}^2(J^h_{-}-h)-{\widetilde G}G_{-}J^h_{-}}\over{\widetilde G G_{-}}}=-h
 \end{align}

Thus, by inserting (\ref{xi1}), (\ref{xi2}) and (\ref{xi3}) in (\ref{JtauMainEquation}), we obtain
$$
d(J^h)^{\tau}={1\over{G_{-}^{\tau}}}d\widehat{M^h} -{{(J^h)^{\tau}_{-}}\over{G_{-}^{\tau}}}d{\widehat m}+
{{J^h-h}\over{\widetilde G}}I_{\Rbrack 0,\tau\Lbrack}dD^{o,\mathbb F}+
{{J^h-h}\over{\widetilde G}}\Delta D^{o,\mathbb F}I_{\Rbrack 0,R\Lbrack}dD   -hI_{\Lbrack R\Rbrack}dD.$$
Thanks to the facts that $J^hI_{\Lbrack R\Rbrack}=0$ (due to $YI_{\Lbrack R\Rbrack}=0$),
$\Rbrack0, R\Lbrack\cap{\Rbrack 0,\tau\Rbrack}={\Rbrack 0,\tau\Lbrack}\cup\left({\Rbrack 0,R\Rbrack} \cap{\Lbrack\tau\Rbrack}\right)$,
and ${{J^h-h}\over{\widetilde G}}\Delta D^{o,\mathbb F}I_{\Rbrack 0,R\Lbrack}dD=
{{J^h-h}\over{\widetilde G}}I_{\Rbrack 0,R\Lbrack}I_{\Lbrack \tau\Rbrack}d D^{o,\mathbb F}$,
 we conclude that the above equality takes the form of
\begin{eqnarray*}\label{Jtaulastequartion}
d(J^h)^{\tau}={1\over{G_{-}^{\tau}}}d\widehat{M^h} -{{(J^h)^{\tau}_{-}}\over{G_{-}^{\tau}}}d{\widehat m}+
(J^h-h)I_{\Rbrack 0,R\Lbrack}{1\over{\widetilde G}}I_{\Rbrack 0,\tau\Rbrack}dD^{o,\mathbb F} -hI_{\Lbrack R\Rbrack}dD.
\end{eqnarray*}
Hence, by combining this with (\ref{processNG}) and (\ref{DecompHexplicit}),
 and using $J^h I_{\Lbrack R\Rbrack}=0$ (since $YI_{\Lbrack R\Rbrack}=0$), we get
\begin{align*}
dH&=(h-J^h)dD+d(J^h)^{\tau}\\
&=(h-J^h)dD+{1\over{G_{-}^{\tau}}}d\widehat{M^h} -{{(J^h)^{\tau}_{-}}\over{G_{-}^{\tau}}}d{\widehat m}+
(J^h-h)I_{\Rbrack 0,R\Lbrack}{1\over{\widetilde G}}I_{\Rbrack 0,\tau\Rbrack}dD^{o,\mathbb F} -hI_{\Lbrack R\Rbrack}dD\\
&={1\over{G_{-}^{\tau}}}d\widehat{M^h} -{{(J^h)^{\tau}_{-}}\over{G_{-}^{\tau}}}d{\widehat m}+(h-J^h)I_{\Rbrack 0,R\Lbrack}dN^{\mathbb G}.
 \end{align*}
This ends the proof of assertion (a).\\
%%%%%%%%%%%%%%%%%%%%%%%%%%%%%%%%%%%%%%%%%%%%%%
%%%%%%%%%
%%%%%%%%%%%%%%%%%%%%%%%%%%%%%%%%%%%%%%%%%%%%
 {\bf Step 2.} To prove assertion (b) it is enough to remark that $H$ is a $\mathbb G$-martingale uniformly integrable, and
  $h\in L^1({\cal O}(\mathbb F),P\otimes D)\subset {\cal I}^o(N^{\mathbb G}, \mathbb G)$. Thus, it is sufficient to prove that
   $J^hI_{\Lbrack 0,R\Lbrack}\in {\cal I}^o(N^{\mathbb G}, \mathbb G)$ when $h\in L\log L({\cal O}(\mathbb F),P\otimes D)$.
   This latter fact requires the following inequality, which holds due to $ J^h_tI_{\Lbrack 0,\tau\Lbrack}(t)=
   \E \big[ h_{\tau} | {\cal G}_t\big]I_{\Lbrack 0,\tau\Lbrack}(t)$,
\begin{align*}
	&\E\left[\int_0^{\infty}\vert J^h_t\vert G_t{\widetilde G_t}^{-1}I_{\Lbrack 0,R\Lbrack}(t)dD_t\right]=
	\E\left[\int_0^{\infty}\vert J^h_t\vert{\widetilde G_t}^{-1}I_{\Lbrack 0,\tau\Lbrack}(t)dD^{o,\mathbb F}_t\right]\\
	&\leq  \E\left[\int_0^{\infty}\E\big[\vert h_{\tau}\vert \mid{\cal G}_t\big]{\widetilde G_t}^{-1}
I_{\Lbrack 0,\tau\Lbrack}(t)dD^{o,\mathbb F}_t\right]=: \E\left[\int_0^{\infty}K^{\mathbb G}_t dV^{\mathbb G}_t\right],
\end{align*}
where
$$K^{\mathbb G}_t:=\E\big[\vert h_{\tau}\vert \mid{\cal G}_t\big] \quad \mbox{and} \quad dV^{\mathbb G}_t:=
\left(\widetilde G_t\right)^{-1}I_{\Lbrack 0,\tau\Lbrack}(t)dD^{o,\mathbb F}_t.$$
Due to $\Delta V^{\mathbb G}=\widetilde G^{-1} \Delta D^{o,\mathbb F} I_{\Lbrack 0,\tau\Lbrack}\leq 1$, we deduce that $$ \E [V^{\mathbb G}_{\infty}-V^{\mathbb G}_{t-}|{\cal G}_t]= \Delta V^{\mathbb G}+  \E [V^{\mathbb G}_{\infty}-V^{\mathbb G}_{t}|{\cal G}_t]\leq 1+\E[V^{\mathbb G}_{\infty}-
V^{\mathbb G}_t|{\cal F}_t](G_t)^{-1}I_{\{t<\tau\}}\leq 1+I_{\{t<\tau\}}\leq 2 .$$ Therefore,  by combining this with Doob's
 inequality for the $\mathbb G$-martingale $K^{\mathbb G}$, we derive
\begin{align*}
	& \E\left[\int_0^{\infty}K^{\mathbb G}_t dV^{\mathbb G}_t\right]\leq \E\left[\int_0^{\infty}\sup_{u\leq t} K^{\mathbb G}_u
 dV^{\mathbb G}_t\right]
	= \E\left[\int_0^{\infty} (V^{\mathbb G}_{\infty}-V^{\mathbb G}_{t-}) d(\sup_{u\leq t} K^{\mathbb G}_u)\right]\\
	&= \E\left[\int_0^{\infty} \E[V^{\mathbb G}_{\infty}-V^{\mathbb G}_{t-}|{\cal G}_t] d(\sup_{u\leq t} K^{\mathbb G}_u)\right]
\leq  2\E\left[\sup_{u\geq 0} K^{\mathbb G}_u\right]\leq 2C \E\left[ K^{\mathbb G}_{\infty}\ln (K^{\mathbb G}_{\infty})\right] + 2C <+\infty,
\end{align*}
where $C$ is a universal positive constant.  This proves that $(h-J^h)I_{\Lbrack 0,R\Lbrack}\is N^{\mathbb G}$
  belongs to ${\cal M}(\mathbb G)$. Hence the remaining process does also belong to ${\cal M}(\mathbb G)$. This ends the proof of assertion (b).\\
  %%%%%%%%%%%%%%%%%%%%%%%%%%%%%%%%%%%%%%%%%%%%%%%
{\bf Step 3.} Herein, we prove assertion (c). To this end, we consider $h\in L^2({\cal O}(\mathbb F), P\otimes D)$, and put
\begin{eqnarray*}
M:=M^h-J^h_{-}\is m,\ \ \ \ \Gamma_n:=\left\{\min(G_{-},\widetilde G)\geq n^{-1}\ \&\ \vert \Delta M\vert\leq n\ \&\
\vert \Delta \left(\Delta M_{\widetilde R}I_{\Rbrack\widetilde{R},+\infty\Rbrack}\right)^{p,\mathbb F}\vert\leq n\right\}.
 \end{eqnarray*}
 Then, from the decomposition (\ref{RepresentationofH}), we calculate
  \begin{eqnarray}
  I_{\Gamma_n}\is [H,H]&=&G_{-}^{-2}I_{\Gamma_n}I_{\Rbrack 0,\tau\Rbrack}\is [\widehat{M},\widehat{M}]+(h-J^h)^2
  I_{\Gamma_n}I_{\Rbrack 0,{\widetilde R}\Lbrack}\is [N^{\mathbb G}, N^{\mathbb G}]\nonumber\\
  &+&2(h-J^h)\left({{\Delta M}\over{\widetilde G}}+G_{-}^{-1}\Delta\left(\Delta M_{\widetilde R}
  I_{\Lbrack\widetilde{R},+\infty\Lbrack}\right)^{p,\mathbb F}\right)
  I_{\Gamma_n}I_{\Rbrack 0,{\widetilde R}\Lbrack}\is N^{\mathbb G}.\label{quadratic}
  \end{eqnarray}
  Since $(h-J^h)I_{\Lbrack 0,{\widetilde R}\Lbrack}\in {\cal I}^o(N^{\mathbb G},\mathbb G)$ and $\left({\widetilde G}^{-1}\Delta M+
  G_{-}^{-1}\Delta\left(\Delta M_{\widetilde R}I_{\Lbrack\widetilde{R},+\infty\Lbrack}\right)^{p,\mathbb F}\right)
  I_{\Gamma_n}$ is $\mathbb F$-optional and bounded, then Theorem \ref{theorem2.13} and Davis' inequality (i.e. $\E[\sup_{0\leq t\leq T}\vert M_t\vert]\leq C\E [M,M]_T^{1/2}$ for any martingale $M$) imply that the last process
  in the RHS term of (\ref{quadratic}) is a uniformly integrable martingale. Thus, on the one hand, we obtain
  \begin{equation}\label{quadraticExpectation}
\displaystyle  \E\Bigl[I_{\Gamma_n}\is [H,H]_{\infty}\Bigr]=\displaystyle
\E\Bigl[G_{-}^{-2}I_{\Gamma_n}I_{\Rbrack 0,\tau\Rbrack}\is [\widehat{M},\widehat{M}]_{\infty}\Big]+
\displaystyle \E\Bigl[(h-J^h)^2
  I_{\Gamma_n}I_{\Rbrack 0,{\widetilde R}\Lbrack}\is [N^{\mathbb G}, N^{\mathbb G}]_{\infty}\Bigr],
  \end{equation}
  and on the other hand, $I_{\Gamma_n}I_{\Rbrack 0,\tau\Rbrack}$ increases to $I_{\Rbrack 0,\tau\Rbrack}$ almost surely.
  Therefore, the proof of assertion (c) follows from a combination of Fatou's lemma, $\E[H,H]_{\infty}<+\infty$, and
  (\ref{quadraticExpectation}).  This proves that the two \(\mathbb G\)-local martingales
 $$L^{\mathbb G}:=\left(hG-M^h +h \is D^{o,\mathbb F}\right)G^{-1}I_{\Rbrack 0,R\Lbrack} \is N^{\mathbb G}\ \mbox{and}\ 
 M^{\mathbb G}:= {1\over{G_-}}I_{\Rbrack 0,\tau\Rbrack} \is  \widehat{M^h}
 -{{M^h_{-} - (h \is D^{o,\mathbb F})_{-}}\over{G_{-}^2}}I_{\Rbrack 0,\tau\Rbrack} \is \widehat{m}$$
 are square integrable. The orthogonality between these martingales  follows from combining the facts that $L^{\mathbb G}$ is a pure default martingale 
 (of the first type) and $M^{\mathbb G}$ takes the form of $M^{\mathbb G}=G_{-}^{-2}I_{\Rbrack0,\tau\Rbrack}\is\widehat M$, where $M\in {\cal M}_{loc}(\mathbb F)$, 
 Remark \ref{orthogonalityofMhat}, and $[L^{\mathbb G}, M^{\mathbb G}]\in{\cal A}(\mathbb G)$. This ends the proof of the theorem.
\end{proof} 

%%%%%XXXXXXXXXXXXXXXXXXXXXXXXXXXXXXX
%%XXXXXXXXX
%%%XXXXXXXXXXXXXXXXXXXXXXXXXXXXXXXXXXX
%%\section{Proof of  Corollary \ref{FullsetofPMM}}\label{AppendixproofCorollary2.21}
%%\qed
%%%%%%%In the proof above, we used frequently the following lemma that we borrow from \cite{aksamitetal15}
%%%%%%%%%%%%%%%%%%%%%%%%%%%%%%%%%%%%%%%%%%%%%%%%%%%%%%%
%%%%%%%%%%%%%%%%%%%%%%%%%%%%%%%%%%%%%%%%%%%%%%%%%%%%%
%%%%%
 %%%%%%%%%%%%%%%%%%%%%%%%%%%%%%%%%%%%%%%%%%%%%%%%%%%%%%%%%%%%%%%
 %%%%%%%%%%%%%%%%%%%%%%%%%%%%%%%%%%%%%%%%%%%%%%%%%%%%%%%%%%%%%%%%%
%%%\section{Proof of Lemma \ref{consequences4MainAssum}}\label{AppendixproofLemma}
%%%%The proof of this lemma requires two lemmas that we start with.

\end{appendices}

\end{document}